\documentclass{scrartcl}
\KOMAoptions{paper=letter}

\usepackage[amsthm]{pamas} 

\usepackage{eucal}
\usepackage{nicefrac}
\usepackage{subfigure}
\usepackage{arydshln}
\usepackage{tabularx, booktabs}
\usepackage{tikz}
\usepackage{xcolor}
\usepackage{framed}
\usepackage{calc}
\usetikzlibrary{arrows}
\usetikzlibrary{decorations.pathreplacing,decorations.markings,snakes,calc}

\newcounter{remark}
\newenvironment{remark}[1][]{\refstepcounter{remark} \ifthenelse{\equal{#1}{}}{\noindent\textbf{Remark~\theremark. }}{\noindent \textbf{Remark~\theremark~(#1).}}}{\medskip}

\title{Consistent Probabilistic Social Choice}

\author{Florian Brandl \quad Felix Brandt \quad Hans Georg Seedig\\
Technische Universit\"at M\"unchen\\\texttt{\small \{brandlfl,brandtf,seedigh\}@in.tum.de}
}

\date{}

\newcommand{\pref}{\succcurlyeq}
\newcommand{\supp}{\mathrm{supp}}
\newcommand{\fone}{\mathcal{F}}

\newcommand{\aff}{\mathrm{aff}}
\newcommand{\conv}{\mathrm{conv}}
\newcommand{\lin}{\mathrm{lin}}
\renewcommand{\int}{\mathrm{int}}

\newcommand{\uni}{\mathrm{uni}}
\newcommand{\cl}{\mathrm{cl}}

\renewcommand{\gcd}{\kappa}
\renewcommand{\epsilon}{\varepsilon}

\newcommand{\ml}[1][]{\ifthenelse{\equal{#1}{}}{\mathit{ML}}{\mathit{ML}(#1)}}
\newcommand{\rd}[1][]{\ifthenelse{\equal{#1}{}}{\mathit{RD}}{\mathit{RD}(#1)}}
\newcommand{\rsd}[1][]{\ifthenelse{\equal{#1}{}}{\mathit{RSD}}{\mathit{RSD}(#1)}}

\colorlet{shadecolor}{orange!25}

\makeatletter
\@ifpackageloaded{subfig}{
		\newcommand{\subfigureCMD}[4][]{\ifthenelse{\equal{#3}{}}{\subfloat[#2]{#4}}{\subfloat[#2\label{#3}]{#4}}}
	}{\relax}
\@ifpackageloaded{subfigure}{
		\newcommand{\subfigureCMD}[4][]{\ifthenelse{\equal{#3}{}}{\subfigure[#2]{#4}}{\subfigure[#2\label{#3}]{#4}}}
	}{\relax}
\@ifpackageloaded{subcaption}{
		\newcommand{\subfigureCMD}[4][]{\ifthenelse{\equal{#3}{}}{\subcaptionbox{#2}{#4}}{\subcaptionbox{#2\label{#3}}{#4}}}
	}{\relax}	
\makeatother

\newcolumntype{P}[1]{>{\centering\arraybackslash}p{#1}}

\newcolumntype{Y}{>{\centering\arraybackslash}X}

\begin{document}

\maketitle

\begin{abstract}
Two fundamental axioms in social choice theory are consistency with respect to a variable electorate and consistency with respect to components of similar alternatives. In the context of traditional non-probabilistic social choice, these axioms are incompatible with each other. 
We show that in the context of \emph{probabilistic} social choice, these axioms uniquely characterize a function proposed by Fishburn (Rev. Econ. Stud., 51(4), 683--692, 1984). 
Fishburn's function returns so-called \emph{maximal lotteries}, \ie lotteries that correspond to optimal mixed strategies in the symmetric zero-sum game induced by the pairwise majority margins. Maximal lotteries are guaranteed to exist due to von Neumann's Minimax Theorem, are almost always unique, and can be efficiently computed using linear programming.
\end{abstract}

\section{Introduction}

Many important properties in the theory of social choice concern the consistency of aggregation functions under varying parameters. 
What happens if two electorates are merged? 
How should an aggregation function deal with components of similar alternatives? 
How should choices from overlapping agendas be related to each other?
These considerations have led to a number of consistency axioms that these functions should ideally satisfy.\footnote{Consistency conditions have found widespread acceptance well beyond social choice theory and feature prominently in the characterizations of various concepts in mathematical economics such as proportional representation rules \citep{BaYo78a}, Nash's bargaining solution \citep{Lens88a}, the Shapley value \citep{HaMa89a}, and Nash equilibrium \citep{PeTi96a}. \citet{Youn94a} and \citet{Thom14a} provide excellent overviews and give further examples.
} 
Unfortunately, social choice theory is rife with impossibility results which have revealed the incompatibility of many of these properties. \citet{YoLe78a}, for example, have pointed out that every social choice function that satisfies Condorcet-consistency violates consistency with respect to variable electorates. On the other hand, it follows from results by \citet{Youn75a} and \citet{Lasl96a} that all Pareto-optimal social choice functions that are consistent with respect to variable electorates are inconsistent with respect to components of similar alternatives.

The main result of this paper is that, in the context of \emph{probabilistic} social choice, consistency with respect to variable electorates and consistency with respect to components of similar alternatives uniquely characterize an appealing probabilistic social choice function, which furthermore satisfies Condorcet-consistency.
Probabilistic social choice functions yield lotteries over alternatives (rather than sets of alternatives) and were first formally studied by \citet{Zeck69a}, \citet{Fish72b}, and \citet{Intr73a}. Perhaps one of the best known results in this context is \citeauthor{Gibb77a}'s characterization of strategyproof probabilistic social choice functions \citep{Gibb77a}. 
An important corollary of Gibbard's characterization, attributed to Hugo Sonnenschein, is that \emph{random dictatorships} are the only strategyproof and \emph{ex post} efficient probabilistic social choice functions. In random dictatorships, one of the voters is picked at random and his most preferred alternative is implemented as the social choice.
While Gibbard's theorem might seem as an extension of classic negative results on strategyproof non-probabilistic social choice functions \citep{Gibb73a,Satt75a}, it is in fact much more positive \citep[see also][]{Barb79b}. In contrast to deterministic dictatorships, the uniform random dictatorship (henceforth, random dictatorship), in which every voter is picked with the same probability, enjoys a high degree of fairness and is in fact used in many subdomains of social choice that are concerned with the fair assignment of objects to agents \citep[see, \eg][]{AbSo98a,BoMo04a,ChKo10a,BCKM12a}.

\subsection*{Summary of Results}

In this paper, we consider two consistency axioms, non-probabilistic versions of which have been widely studied in the literature. The first one, \emph{population-consistency}, requires that, whenever two disjoint electorates agree on a lottery, this lottery should also be chosen by the union of both electorates.
The second axiom, \emph{composition-consistency}, requires that the probability that an alternative receives is unaffected by introducing new variants of another alternative. Alternatives are variants of each other if they form a component, \ie they bear the same relationship to all other alternatives.
On top of that, composition-consistency prescribes that the probability of an alternative within a component should be directly proportional to the probability that the alternative receives when the component is considered in isolation.
Apart from their intuitive appeal, these axioms can be motivated by the desire to 
prevent a central planner from strategically partitioning the electorate into subelectorates or by deliberately introducing similar variants of alternatives, respectively.
Our first result shows that there is no non-probabilistic social choice function that satisfies both axioms simultaneously. We then move to probabilistic social choice and prove that the only probabilistic social choice function satisfying these properties is the function that returns all \emph{maximal lotteries} for a given preference profile. 
Maximal lotteries, which were proposed independently by \citet{Fish84a} and other authors, are equivalent to mixed maximin strategies of the symmetric zero-sum game given by the pairwise majority margins. Whenever there is an alternative that is preferred to any other alternative by some majority of voters (a so-called Condorcet winner), the lottery that assigns probability~1 to this alternative is the unique maximal lottery. In other words, maximal lotteries satisfy \emph{Condorcet-consistency}.
At the same time, maximal lotteries satisfy population-consistency which has been identified by \citet{Youn74a}, \citet{NiRu81a}, \citet{Saar90b}, and others as \emph{the} defining property of Borda's scoring rule. As such, the characterization can be seen as one possible resolution of the historic dispute between the founding fathers of social choice theory, the Chevalier de Borda and the Marquis de Condorcet, which dates back to the 18th century.\footnote{In this sense, our main theorem is akin to the characterization of Kemeny's rule by \citet{YoLe78a}. \citeauthor{YoLe78a} considered \emph{social preference functions}, \ie functions that return sets of \emph{rankings} of alternatives, and showed that Kemeny's rule is characterized by strong versions of population-consistency and Condorcet-consistency.\\
Interestingly, all three rules---Borda's rule, Kemeny's rule, and maximal lotteries---maximize aggregate score in a well-defined sense. For maximal lotteries, this is the case because they maximize social welfare according to the canonical skew-symmetric bilinear (SSB) utility functions representing the voters’ ordinal preferences \citep{BBH15c}. SSB utility theory goes back to \citet{Fish82c} and is a generalization of von Neumann-Morgenstern utility theory.}

Random dictatorship, the most prevalent probabilistic social choice function, satisfies population-consistency, but fails to satisfy composition-consistency. For this reason, we also consider a weaker version of composition-consistency called cloning-consistency, which is satisfied by random dictatorship, and provide an alternative characterization of maximal lotteries using population-consistency, cloning-consistency, and Condorcet-consistency (see \remref{rem:cloning}). 

\subsection*{Acceptability of Social Choice Lotteries}

Clearly, allowing lotteries as outcomes for high-stakes political elections such as those for the U.S.~presidency would be highly controversial and considered by many a failure of deliberative democracy. 
If, on the other hand, a small group of people repeatedly votes on where to hold their next meeting, randomization would likely be more acceptable and perhaps even desirable.
The use of lotteries for the selection of officials interestingly goes back to the world's first democracy in Athens where it was widely regarded as a principal characteristic of democracy \citep{Head33a}.
It has also been early observed in the social choice literature that ``unattractive social choices may result whenever lotteries are not allowed to compete. [\dots] Refusal to entertain lotteries on alternatives can lead to outcomes that to many appear to be inequitable and perhaps even inefficient'' \citep{Zeck69a}.\footnote{It is interesting to note that the ``intransitivity difficulties'' that \citet{Zeck69a} examines in the context of probabilistic social choice disappear when replacing majority rule with \emph{expected} majority rule. This directly leads to Fishburn's definition of maximal lotteries.}
 In contemporary research, probabilistic social choice has gained increasing interest in both social choice \citep[see, \eg][]{EPS02a,BMS05a,CSZ14a} and political science \citep[see, \eg][]{Good05a,Dowl09a,Ston11a}.

Whether lotteries are socially acceptable depends on many factors, only some of which are based on formal arguments. In our view, two important factors are \emph{(i) the effective degree of randomness} and \emph{(ii) risk aversion on behalf of the voters}. 

As to \emph{(i)}, it is easily seen that certain cases call for randomization or other means of tie-breaking. For example, if there are two alternatives, $a$ and $b$, and exactly half of the voters prefer $a$ while the other half prefers $b$, there is no deterministic way of selecting a single alternative without violating basic fairness conditions. There are several possibilities to extend the notion of a tied outcome to three or more alternatives. An important question in this context is whether ties are a rare exception or a common phenomenon.
A particularly rigorous and influential extension due to Condorcet declares a tie in the absence of a pairwise majority winner. According to Condorcet, it is the intransitivity of social preferences, as exhibited in the Condorcet paradox, that leads to situations in which there is no unequivocal winner. As it turns out, maximal lotteries are degenerate if and only if there is a Condorcet winner. Our main result thus establishes that the degree of randomness entailed by our axioms is precisely in line with Condorcet's view of equivocality.
Interestingly, there is strong empirical and experimental evidence that most real-world preference profiles for \emph{political} elections do admit a Condorcet winner \citep[see, \eg][]{RGMT06a,Lasl10a,GeLe11a}.\footnote{
Analytical results for the likelihood of Condorcet winners are typically based on the simplistic ``impartial culture'' model, which assumes that every preference relation is equally likely. According to this model, a Condorcet winner, for example, exists with a probability of at least 63\% when there are seven alternatives \citep{Fish73a}. 
The impartial culture model is considered highly unrealistic and \citet{RGMT06a} argued that it significantly underestimates the probability of Condorcet winners. \citet{GeLe11a} summarized 37 empirical studies from 1955 to 2009 and concluded that ``there is a possibility that Condorcet's Paradox might be observed, but that it probably is not a widespread phenomenon.''
\citet{Lasl10a} and \citet{BrSe14a} reported concrete probabilities for the existence of Condorcet winners under various distributional assumptions using computer simulations.
A common observation in these studies is that the probability of a Condorcet winner generally decreases with increasing number of alternatives.
For example, \citeauthor{BrSe14a} found that, for 15 voters and a spatial distribution of preferences that is commonly used in political science, the probability of a Condorcet winner ranges from 98\% (for three alternatives) to 59\% (for 50 alternatives).
} Maximal lotteries only randomize in the less likely case of cyclical majorities. \citet{BrSe14a} specifically analyzed the support of maximal lotteries and found that the average support size is less than four under various distributional assumptions and up to 30 alternatives.
By contrast, random dictatorship randomizes over all alternatives in almost all elections.

As to \emph{(ii)}, risk aversion is strongly related to the frequency of preference aggregation.
If an aggregation procedure is not frequently repeated, the law of large numbers does not apply and risk-averse voters might prefer a sure outcome to a lottery whose expectation they actually prefer to the sure outcome. Hence, probabilistic social choice seems particularly suitable for novel preference aggregation settings that have been made possible by technological advance. The Internet, in particular, allows for much more frequent preference aggregation than traditional paper-and-pencil elections. 
In recurring randomized elections with a fixed set of alternatives, voters need not resubmit their preferences for every election; rather preferences can be stored and only changed if desired. For example, maximal lotteries could help a group of coworkers with the daily decision of where to have lunch without requiring them to submit their preferences every day. 
Another example are automatic music broadcasting systems, such as Internet radio stations or software DJs, that decide which song should be played next based on the preferences of the listeners.
In contrast to traditional deterministic solutions to these problems such as sequential dictatorships, the repeated execution of lotteries is a memoryless process that guarantees \emph{ex ante} fairness after any number of elections.

Finally, it should be noted that the lotteries returned by probabilistic social choice functions do not necessarily have to be interpreted as probability distributions. They can, for instance, also be seen as fractional allocations of divisible objects such as time shares or monetary budgets. The axioms considered in this paper are equally natural for these interpretations as they are for the probabilistic interpretation.

\section{Preliminaries}
\label{sec:prelims}

Let $U$ be an infinite universal set of alternatives.
The set of \emph{agendas} from which alternatives are to be chosen is the set of finite and non-empty subsets of~$U$, denoted by~$\fone(U)$.
The set of all linear (\ie complete, transitive, and antisymmetric) \emph{preference relations} over some set $A\in\fone(U)$ will be denoted by $\mathcal{L}(A)$.

For some finite set $X$, we denote by $\Delta(X)$ the set of all probability distributions with rational values over $X$.
A (fractional) preference profile $R$ for a given agenda $A$ is an element of $\Delta(\mathcal{L}(A))$, which can be associated with the $(|A|!-1)$-dimensional (rational) unit simplex.
We interpret $R({\pref})$ as the fraction of voters with preference relation ${\pref}\in\mathcal{L}(A)$. Preference profiles are depicted by tables in which each column represents a preference relation $\pref$ with $R({\pref})>0$.
The table below shows an example profile on three alternatives.\footnote{Our representation of preference profiles implicitly assumes that aggregation functions are anonymous (\ie all voters are treated identically) and homogeneous (\ie duplication of the electorate does not affect the outcome). 
Similar models (sometimes even assuming a continuum of voters) have for example been considered by \citet{Youn74b,Youn75a}, \citet{YoLe78a}, \citet{Saar95a}, \citet{DaMa08a}, \citet{ChKo10a}, and \citet{BuCa12a}.}
\[
\begin{array}{ccc}
\nicefrac{1}{2} & \nicefrac{1}{3} & \nicefrac{1}{6} \\\midrule
a & a & b \\
b & c & c \\
c & b & a
\end{array}
\tag{Example 1}
\]
The set of all preference profiles for a fixed agenda $A$ is denoted by $\mathcal{R}|_A$ and  $\mathcal{R}$ is defined as the set of all preference profiles, \ie $\mathcal{R} = \bigcup_{A\in\fone(U)} \mathcal{R}|_A$.
For $B\subseteq A$ and $R\in\mathcal{R}|_A$, $R|_B$ is the restriction of $R$ to alternatives in $B$, \ie for all ${\pref}\in\mathcal{L}(B)$,

\[R|_B({\pref}) = \sum_{{\pref}'\in \mathcal{L}(A)\colon{\pref}\subseteq{\pref}'} R({\pref}')\text.\]

For all $x,y\in A$, $R(x,y) = R|_{\{x,y\}}({\{(x,y)\}})$ is the fraction of voters who prefer $x$ to $y$ (the set $\{(x,y)\}$ represents the preference relation on two alternatives with $x\succ y$). In Example~1, $R(a,b) = \nicefrac{5}{6}$.

Elements of $\Delta(A)$ are called \emph{lotteries} and will be written as convex combinations of alternatives.
If $p$ is a lottery, $p_x$ is the probability that $p$ assigns to alternative $x$.

A \emph{probabilistic social choice function (PSCF)} $f$ is an (upper hemi-) continuous function that, for any agenda $A\in\mathcal{F}(U)$, maps a preference profile $R\in \mathcal{R}|_A$ to a non-empty convex subset of $\Delta(A)$.\footnote{\citet[][pp.~248--249]{Fish73a} argued that the set of lotteries returned by a probabilistic social choice function should be convex because it would be unnatural if two lotteries are socially acceptable while a randomization between them is not \citep[see also][p.~201]{Fish72b}.} 
A PSCF is thus a collection of mappings from high-dimensional simplices to low-dimensional simplices.
Two further properties that we demand from any PSCF are \emph{unanimity} and \emph{decisiveness}. Unanimity states that in the case of one voter and two alternatives, the preferred alternative should be chosen with probability~1.\footnote{This is the only condition we impose that actually interprets the preference relations. It is equivalent to \emph{ex post} efficiency for agendas of size~2 and is slightly weaker than Young's faithfulness \citep{Youn74a}.
Our results still hold when replacing unanimity with the even less controversial condition that merely requires that $f(R)\ne\{y\}$ whenever $R(x,y)=1$.} Since we only consider fractional preference profiles, this amounts to for all $x,y\in U$ and $R\in\mathcal{R}|_{\{x,y\}}$,
\begin{align*}
f(R) = \{x\} \text{ whenever } R(x,y) = 1\text{.}\tag{unanimity}
\intertext{Decisiveness requires that the set of preference profiles where $f$ is multi-valued is negligible in the sense that for all $A\in \mathcal{F}(U)$,}
\{R\in\mathcal{R}|_A \colon |f(R)|=1\} \text{ is dense in } \mathcal{R}|_A\text{.}\tag{decisiveness}
\end{align*}
In other words, for every preference profile that yields multiple lotteries, there is an arbitrarily close preference profile that only yields a single lottery.

Probabilistic social choice functions considered in the literature usually satisfy these conditions and are therefore well-defined PSCFs. For example, consider \emph{random dictatorship ($\rd$)}, in which one voter is picked uniformly at random and his most-preferred alternative is returned. Formally, $\rd$ returns the unique lottery, which is determined by multiplying fractions of voters with their respective top choices, \ie for all $A\in\mathcal{F}(U)$ and $R\in\mathcal{R}|_A$,
\[
\rd(R) = \left\{\sum_{{\pref}\in \mathcal{L}(A)} R(\pref) \cdot \max_\pref(A)\right\}\text{,}\tag{random dictatorship}
\]
where $\max_\pref(A)$ denotes the unique alternative $x$ such that $x\pref y$ for all $y\in A$.
For the preference profile $R$ given in Example~1,
\[
\rd[R] = \{\nicefrac{5}{6}\,a+\nicefrac{1}{6}\,b\}\text{.}
\]
$\rd$ is single-valued and therefore trivially decisive and convex-valued.
It is also easily seen that $\rd$ satisfies unanimity and continuity and thus constitutes a PSCF.

A useful feature of our definition of PSCFs is that traditional \emph{set-valued} social choice functions (SCFs) (also known as \emph{social choice correspondences}) can be seen as a special case, namely as those PSCFs that map every preference profile $R\in\mathcal{R}|_A$ on some agenda $A\in\mathcal{F}(U)$ to $\Delta(X)$ for some $X\subseteq A$.
Such PSCFs will be called \emph{(non-probabilistic) SCFs}.

\section{Population-consistency and Composition-consistency}
\label{sec:consistency}

The consistency conditions we consider are generalizations of the corresponding conditions for SCFs, \ie the axioms coincide with their non-probabilistic counterparts.

The first axiom relates choices from varying electorates to each other. More precisely, it requires that whenever a lottery is chosen simultaneously by two electorates, this lottery is also chosen by the union of both electorates. For example,
consider the two preference profiles $R'$ and $R''$ given below.
\[
\begin{array}{cc}
\nicefrac{1}{2} & \nicefrac{1}{2} \\\midrule
a & b \\
b & c \\
c & a \\[2ex]
\multicolumn{2}{c}{R'}\\
\end{array}
\qquad\qquad
\begin{array}{cc}
\nicefrac{1}{2} & \nicefrac{1}{2} \\\midrule
a & b \\
c & c \\
b & a \\[2ex]
\multicolumn{2}{c}{R''}\\
\end{array}
\qquad\qquad
\begin{array}{lcccr}
&\nicefrac{1}{4} & \nicefrac{1}{4} & \nicefrac{1}{2} & \\\cmidrule{2-4}
&a & a & b &\\
&b & c & c &\\
&c & b & a &\\[2ex]
\multicolumn{5}{c}{\nicefrac{1}{2}\, R' + \nicefrac{1}{2}\, R''}\\
\end{array}
\tag{Example 2}
\]
Population-consistency then demands that any lottery that is chosen in both $R'$ and $R''$ (say, $\nicefrac{1}{2}\, a + \nicefrac{1}{2}\, b$) also has to be chosen when both preference profiles are merged.
Formally, a PSCF satisfies population-consistency if for all $A\in\mathcal{F}(U)$, $R',R''\in\mathcal{R}|_A$, and any convex combination $R$ of $R'$ and $R''$, \ie $R=\lambda R' +(1-\lambda) R''$ for some $\lambda\in [0,1]$,

\[f(R') \cap f(R'') \quad\subseteq\quad 
f(R)\text{.}
\tag{population-consistency}\]

Population-consistency is arguably one of the most natural axioms for variable electorates and is usually considered in a slightly stronger version, known as \emph{reinforcement} or simply \emph{consistency}, where the inclusion in the equation above is replaced with equality whenever the left-hand side is non-empty (see also \remref{rem:pcons}).
Note that population-consistency is merely a statement about abstract sets of outcomes, which makes no reference to lotteries whatsoever.
 It was first considered independently by \citet{Smit73a}, \citet{Youn74a}, and \citet{FiFi74a} and features prominently in the characterization of scoring rules by \citet{Smit73a} and \citet{Youn75a}. Population-consistency and its variants have found widespread acceptance in the social choice literature \citep[see, \eg][]{Youn74b,Fish78d,YoLe78a,Saar90a,Saar95a,Myer95b,CoMe12a}.

The second axiom prescribes how PSCFs should deal with \emph{decomposable} preference profiles. 
For two agendas $A,B\in\fone(U)$, $B\subseteq A$ is a component in $R\in\mathcal{R}|_{A}$ if the alternatives in $B$ are \emph{adjacent} in all preference relations that appear in $R$, \ie for all $a\in A \setminus B$ and $b,b'\in B$, $a \pref b$ if and only if $a \pref b'$ for all ${\pref}\in\mathcal{L}(A)$ with $R({\pref})>0$. Intuitively, the alternatives in $B$ can be seen as variants or clones of the same alternative because they have exactly the same relationship to all alternatives that are not in $B$. For example, consider the following preference profile $R$ in which $B=\{b,b'\}$ constitutes a component. 

\[
\begin{array}{ccc}
\nicefrac{1}{3} & \nicefrac{1}{6} & \nicefrac{1}{2} \\\midrule
a & a & b \\
b' & b & b' \\
b & b' & a \\[2ex]
\multicolumn{3}{c}{R}\\
\end{array}
\qquad\qquad
\begin{array}{cc}
\nicefrac{1}{2} & \nicefrac{1}{2} \\\midrule
a & b \\
b & a \\ \\[2ex]
\multicolumn{2}{c}{R|_{A'}}\\
\end{array}
\qquad\qquad
\begin{array}{cc}
\nicefrac{1}{3} & \nicefrac{2}{3} \\\midrule
b' & b \\
b & b' \\ \\[2ex]
\multicolumn{2}{c}{R|_B}\\ 
\end{array}
\tag{Example 3}
\]
The `essence' of $R$ is captured by $R|_{A'}$, where $A'=\{a,b\}$ contains only one of the cloned alternatives. It seems reasonable to demand that a PSCF should assign the same probability to $a$ (say, $\nicefrac{1}{2}$) independently of the number of clones of $b$ and the internal relationship between these clones. This condition will be called cloning-consistency and was first proposed by \citet{Tide87a} \citep[see also][]{ZaTi89a}. Its origins can be traced back to earlier, more general, decision-theoretic work by \citet{ArHu72a} and \citet{Mask79a} where it is called \emph{deletion of repetitious states} as well as early work on majoritarian SCFs by \citet{Moul86a}.
For a formal definition of cloning-consistency, let $A',B\in\fone(U)$ and $A = A'\cup B$ such that $A'\cap B = \{b\}$. Then, a PSCF $f$ satisfies cloning-consistency if, for all $R\in\mathcal{R}|_A$ such that $B$ is a component in $R$, 
\[\left\{(p_x)_{x\in A\setminus B}\colon p\in f(R)\right\} = \left\{(p_x)_{x\in A\setminus B}\colon p\in f(R|_{A'})\right\}\text{.}\]

When having a second look at Example~3, it may appear strange that cloning-consistency does not impose any restrictions on the probabilities that $f$ assigns to the clones. While clones behave completely identical with respect to uncloned alternatives, they are not indistinguishable from \emph{each other}. It seems that the relationships between clones ($R|_B$) should be taken into account as well. For example, one would expect that $f$ assigns more probability to $b$ than to $b'$ because two thirds of the voters prefer $b$ to $b'$. An elegant and mathematically appealing way to formalize this intuition is to require that the probabilities of the clones $b$ and $b'$ are directly proportional to the probabilities that $f$ assigns to these alternatives when restricting the preference profile to the component $\{b,b'\}$. This condition, known as \emph{composition-consistency}, is due to \citet{LLL96a} and was studied in detail for majoritarian SCFs \citep[see, \eg][]{Lasl96a,Lasl97a,Bran11b,BBS11a,Hora13a}.\footnote{More generally, modular decompositions of discrete structures have found widespread applications in operations research and combinatorial optimization \citep[see, \eg][]{Moeh85a}.}

For a formal definition of composition-consistency, let $p\in\Delta(A')$ and $q\in\Delta(B)$ and define
	\[(p\times_b q)_x =
	\begin{cases}
		p_x &\text{if } x\in A\setminus B,\\
		p_b q_x\quad &\text{if } x \in B.
	\end{cases}\]
	The operator $\times_b$ is extended to sets of lotteries $X\subseteq\Delta(A')$ and $Y\subseteq\Delta(B)$ by applying it to all pairs of lotteries in $X\times Y$, \ie
	$X\times_b Y = \{p\times_b q\in\Delta(A)\colon p\in X \text{ and } q\in Y\}$.
	
	Then, a PSCF $f$ satisfies composition-consistency if for all $R\in\mathcal{R}|_A$ such that $B$ is a component in $R$, 
\[f(R|_{A'})\times_b f(R|_B) \quad=\quad f(R) \text{.} \tag{composition-consistency}\]
In Example~3 above, $\nicefrac{1}{2}\, a+\nicefrac{1}{2}\, b\in f(R|_{A'})$,  $\nicefrac{2}{3}\, b+\nicefrac{1}{3}\, b'\in f(R|_B)$, and composition-consistency would imply that $\nicefrac{1}{2}\, a+\nicefrac{1}{3}\, b+\nicefrac{1}{6}\, b'\in f(R)$.

\section{Non-probabilistic Social Choice}

In the context of SCFs (\ie PSCFs that only return the convex hull of degenerate lotteries), there is some friction between population-consistency and composition-consistency. 
In fact, the conflict between these notions can be traced back to the well-documented dispute between the pioneers of social choice theory, the Chevalier de Borda and the Marquis de Condorcet
\citep[see, \eg][]{Blac58a,Youn88a,Youn95a,McHe94a}. Borda proposed a score-based voting rule---\emph{Borda's rule}---that can be axiomatically characterized using population-consistency \citep{Youn74a}. It then turned out that the entire class of scoring rules (which apart from Borda's rule also includes plurality rule) is characterized by population-consistency \citep{Smit73a,Youn74b,Youn75a}. Condorcet, on the other hand, advocated \emph{Condorcet-consistency}, which requires that an SCF selects a Condorcet winner whenever one exists. As Condorcet already pointed out, Borda's rule fails to be Condorcet-consistent.  Worse, \citet{YoLe78a} even showed that no Condorcet-consistent SCF satisfies population-consistency (the defining property of scoring rules).\footnote{Theorem 2 by \citet{YoLe78a} actually uses the strong variant of population-consistency, but their proof also holds for population-consistency as defined in this paper.} 
\citet{Lasl96a}, on the other hand, showed that no Pareto-optimal rank-based rule---a generalization of scoring rules---satisfies composition-consistency while this property is satisfied by various Condorcet-consistent SCFs \citep{LLL96a}.
One of the few SCFs that satisfies both properties is the rather indecisive Pareto rule (which returns all alternatives that are not Pareto-dominated). Since our definition of PSCFs already incorporates a certain degree of decisiveness, we obtain the following impossibility. (The proofs of all theorems are deferred to the Appendix.)

\begin{theorem}\label{thm:scf}
	There is no SCF that satisfies population-consistency and composition-consistency.\footnote{\thmref{thm:scf} still holds when replacing composition-consistency with the weaker condition of cloning-consistency.}
\end{theorem}

In light of this result, it is perhaps surprising that, for probabilistic social choice, both axioms are not only mutually compatible but even \emph{uniquely} characterize a PSCF.

\section{Characterization of Maximal Lotteries}
\label{sec:ml}

Maximal lotteries were first considered by \citet{Krew65a} and independently proposed and studied in more detail by \citet{Fish84a}. Interestingly, maximal lotteries or variants thereof have been rediscovered again by economists \citep{LLL93b},\footnote{\citet{LLL93b,LLL96a}, \citet{DuLa99a}, and others have explored the \emph{support} of maximal lotteries, called the \emph{bipartisan set} or the \emph{essential set}, in some detail.  
\citet{Lasl00a} has characterized the essential set using monotonicity, Fishburn's C2, regularity, inclusion-minimality, the strong superset property, and a variant of composition-consistency.} 
mathematicians \citep{FiRy95a}, political scientists \citep{FeMa92a}, and computer scientists \citep{RiSh10a}.\footnote{\citet{FeMa92a} and \citet{RiSh10a} also discussed whether maximal lotteries are suitable for real-world political elections. \citeauthor{RiSh10a} concluded that
``[the maximal lotteries system] is not only theoretically interesting and optimal, but simple to use in practice; it is probably easier to implement than, say, IRV [instant-runoff voting]. We feel that it can be recommended for practical use.''  \citeauthor{FeMa92a} wrote that ``an inherent special feature of [maximal lotteries] is its extensive and essential reliance on probability in selecting the winner [\dots] Without sufficient empirical evidence it is impossible to say whether this feature of [maximal lotteries] makes it socially less acceptable than other majoritarian procedures. It is not at all a question of fairness, for nothing could be fairer than the use of lottery as prescribed by [maximal lotteries]. The problem is whether society will accept such an extensive reliance on chance in public decision-making. Different societies may have differing views about this. For example, it is well known that the free men of ancient Athens regarded it as quite acceptable to select holders of public office by lot. Clearly, before [the maximal lotteries system] can be applied in practice, public opinion must first be consulted, and perhaps educated, on this issue.''
}

In order to define maximal lotteries, we need some notation. For $A\in \mathcal{F}(U)$, $R \in \mathcal{R}|_A$, $x,y \in A$, the entries $M_R(x,y)$ of the \emph{majority margin matrix} $M_R$ denote the difference between the fraction of voters who prefer~$x$ to~$y$ and the fraction of voters who prefer~$y$ to~$x$, \ie
\[M_R(x,y)= R(x,y) - R(y,x)\text.\]
Thus, $M_R$ is skew-symmetric and $M_R\in [-1,1]^{A\times A}$. 
A \emph{(weak) Condorcet winner} is an alternative $x$ that is \emph{maximal} in $A$ according to $M_R$ in the sense that $M_R(x,y) \ge 0$ for all $y \in A$. If $M_R(x,y)>0$ for all $y\in A\setminus\{x\}$, $x$ is called a \emph{strict Condorcet winner}.
A PSCF is \emph{Condorcet-consistent} if $x\in f(R)$ whenever $x$ is a Condorcet winner in $R$.

It is well known from the Condorcet paradox that maximal elements may fail to exist. 
As shown by \citet{Krew65a} and \citet{Fish84a}, this drawback can, however, be remedied by considering lotteries over alternatives. For two lotteries $p,q \in \Delta(A)$, the majority margin can be extended to its bilinear form $p^T M_R q$, the \emph{expected} majority margin.
The set of maximal lotteries is then defined as the set of ``probabilistic Condorcet winners.'' Formally, for all $A\in\mathcal{F}(U)$ and $R\in\mathcal{R}|_A$,\footnote{Several authors apply the signum function to the entries of $M_R$ before computing maximal lotteries. This is, for example, the case for \citet{Krew65a}, \citet{FeMa92a}, \citet{LLL93a}, and \citet{FiRy95a}. Maximal lotteries as defined in this paper were considered by \citet{DuLa99a}, \citet{Lasl00a}, and \citet{RiSh10a}. \citet{Fish84a} allowed the application of any odd function to the entries of $M_R$, which covers both variants as special cases.\label{fn:c1c2}}
\[
\ml[R]=\{p\in \Delta(A) \colon p^T M_R q\ge 0 \text{ for all } q\in \Delta(A)\}\text.
\tag{maximal lotteries}
\]
As an example, consider the preference profile given in Example~1 of \secref{sec:prelims}. Alternative $a$ is a strict Condorcet winner and $\ml(R) = \{a\}$.
This is in contrast to $\rd(R) = \{\nicefrac{5}{6}\,a+\nicefrac{1}{6}\,b\}$, which puts positive probability on any first-ranked alternative no matter how small the corresponding fraction of voters.

The Minimax Theorem implies that $\ml[R]$ is non-empty for all $R \in \mathcal{R}$ \citep{vNeu28a}. In fact, $M_R$ can be interpreted as the payoff matrix of a symmetric zero-sum game and maximal lotteries as the mixed maximin strategies (or Nash equilibrium strategies) of this game. Hence, maximal lotteries can be efficiently computed via linear programming. Interestingly, $\ml(R)$ is a singleton in almost all cases. In particular, this holds if there is an odd number of voters \citep{LLL97a,LeBr05a}. Moreover, we point out in Appendix~\ref{sec:mlsat} that the set of preference profiles that yield a unique maximal lottery is open and dense, which implies that the set of profiles with multiple maximal lotteries is nowhere dense and thus negligible. As a consequence, $\ml$ satisfies decisiveness as well as the other properties we demand from a PSCF (such as unanimity, continuity, and convex-valuedness) and therefore constitutes a well-defined PSCF.

In contrast to the non-probabilistic case where majority rule is known as the only reasonable and fair SCF on two alternatives \citep{May52a,DaMa08a}, there is an infinite number of such PSCFs even when restricting attention to only two alternatives \citep[see, \eg][]{Saun10a}. Within our framework of fractional preference profiles, a PSCF on two alternatives can be seen as a convex-valued continuous correspondence from the unit interval to itself. Unanimity fixes the function values at the endpoints of the unit interval, decisiveness requires that the points where the function is multi-valued are isolated, and population-consistency implies that the function is monotonic.
Two natural extreme cases of functions that meet these requirements are a probabilistic version of simple majority rule and the proportional lottery (\figref{fig:twoalts}).\footnote{\citet{FiGe77a} compared these two-alternative PSCFs on the basis of expected voter satisfaction and found that the simple majority rule outperforms the proportional rule.} Interestingly, these two extreme points are taken by maximal lotteries and random dictatorship as for all $x,y\in U$ and $R\in\mathcal{R}|_{\{x,y\}}$,
\[
\ml(R) = 
\begin{cases}
\{x\} & \text{if }R(x,y)>\nicefrac{1}{2}\text{,} \\
\{y\} & \text{if }R(x,y)<\nicefrac{1}{2}\text{,} \\
\Delta(\{x,y\}) & \text{otherwise,}
\end{cases}
\qquad\text{and}\qquad
	\rd(R) = \{R(x,y)\, x + R(y,x)\,y\}\text.
\]

\begin{figure}[htb]
  \centering
	\subfigureCMD{Maximal lotteries}{fig:ml}{
	\begin{tikzpicture}[scale=3.1415926]
\draw[gray,very thin,step=1] (0,0) grid (1,1);
\draw (0,0) -- (1,0);
  \foreach \x/\xtext in {0/0, 0.5/\nicefrac{1}{2}, 1/1}
    \draw[shift={(\x,0)}] (0pt,1pt) -- (0pt,-1pt) node[below] {$\xtext$};
\draw (0,0) -- (0,1);
  \foreach \y/\ytext in {0/0, 1/1}
    \draw[shift={(0,\y)}] (1pt,0pt) -- (-1pt,0pt) node[left] {$\ytext$};
\draw (0.5,-0.35) node {$R(x,y)$};
  \draw (-0.1,0.5) node[left] {$p_x$};
\draw[shorten >=-.1pt,very thick,rounded corners=0.3, butt cap-butt cap] (0,0) -- (0.5,0) -- (0.5,1) -- (1,1);
\end{tikzpicture}
\label{subfig:ml}
}
\qquad\qquad
	\subfigureCMD{Random dictatorship}{fig:rd}{
\begin{tikzpicture}[scale=3.1415926]
\draw[gray,very thin,step=1] (0,0) grid (1,1);
\draw (0,0) -- (1,0);
  \foreach \x/\xtext in {0/0, 0.5/\nicefrac{1}{2}, 1/1}
    \draw[shift={(\x,0)}] (0pt,1pt) -- (0pt,-1pt) node[below] {$\xtext$};
\draw (0,0) -- (0,1);
  \foreach \y/\ytext in {0/0, 1/1}
    \draw[shift={(0,\y)}] (1pt,0pt) -- (-1pt,0pt) node[left] {$\ytext$};
\draw (0.5,-0.35) node {$R(x,y)$};
  \draw (-0.1,0.5) node[left] {$p_x$};
\draw[shorten >=-.15pt,very thick,triangle 90 cap-triangle 90 cap] (0,0) -- (1,1);
\end{tikzpicture}
\label{subfig:rd}
}
\caption{Maximal lotteries and random dictatorship on two-element agendas.}
  \label{fig:twoalts}
\end{figure}
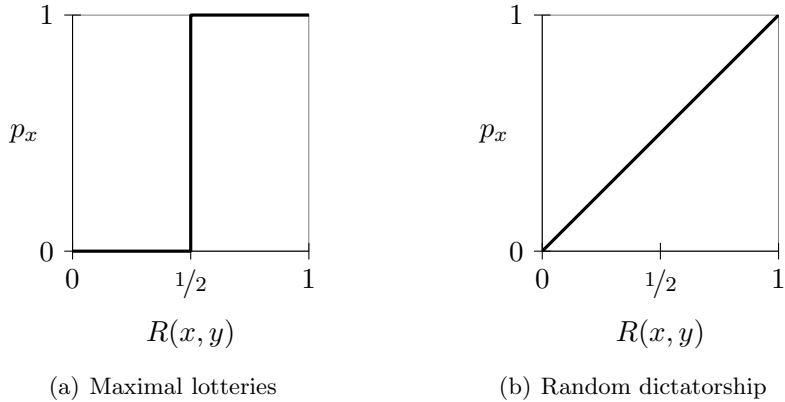

When considering up to three alternatives and additionally taking composition-consistency into account, any such PSCF coincides with majority rule on two-element agendas.\footnote{This also shows that $\rd$ violates composition-consistency, which can be seen in Example~3 in \secref{sec:consistency}. However, $\rd$ does satisfy cloning-consistency (see \remref{rem:cloning}).}
When allowing an arbitrary number of alternatives, the axioms completely characterize $\ml$.

\begin{theorem}\label{thm:fisml}
A PSCF $f$ satisfies population-consistency and composition-consistency if and only if $f=\ml$.
\end{theorem}
The proof of \thmref{thm:fisml} is rather involved yet quite instructive as it rests on a number of lemmas that might be of independent interest (see Appendix~\ref{sec:fml}). The high-level structure is as follows. The fact that $\ml$ satisfies population-consistency and composition-consistency follows relatively easily from basic linear algebra. 

For the converse direction, we first show that population-consistency and composition-consistency characterize $\ml$ on two-element agendas. 
For agendas of more than two alternatives, we assume that $f$ is a population-consistent and composition-consistent PSCF and then show that $f\subseteq \ml$ and $\ml \subseteq f$. The first statement takes up the bulk of the proof and is shown by assuming for contradiction that there is a preference profile for which $f$ yields a lottery that is \emph{not} maximal. We then identify a set of preference profiles with full dimension for which $f$ returns the uniform lottery over a fixed subset of at least two alternatives and which has the uniform profile, \ie the preference profile in which every preference relation is assigned the same fraction of voters, in its interior. Along the way we show that $f$ has to be Condorcet-consistent for all preference profiles that are close to the uniform profile.
It follows that there has to be an $\epsilon$-ball of profiles around some strict Condorcet profile (close to the uniform profile), for which $f$ returns the uniform lottery over a non-singleton subset of alternatives as well as the lottery with probability~1 on the Condorcet winner. This contradicts decisiveness.
For the inclusion of $\ml$ in $f$, we take an arbitrary preference profile and an arbitrary vertex of the set of maximal lotteries for this profile and then construct a sequence of preference profiles that converges to the original profile and whose maximal lotteries converge to the specified maximal lottery. From $f\subseteq\ml$ and continuity, we obtain that $f$ has to select this lottery in the original preference profile. Finally, convexity implies that $\ml\subseteq f$.

\section{Remarks}

We conclude the paper with a number of remarks.

\medskip
\begin{remark}[Independence of axioms]\label{rem:ind}
The axioms used in \thmref{thm:fisml} are independent from each other.
$\rd$ satisfies population-consistency, but violates composition-consistency (see also \remref{rem:cloning}). The same is true for Borda's rule.
When defining $\ml^3$ via the third power of majority margins $(M_R(x,y))^3$, $\ml^3$ satisfies composition-consistency, but violates population-consistency.\footnote{Such variants of $\ml$ were already considered by \citet{Fish84a}. See also Footnote~\ref{fn:c1c2}.} 
Also the properties implicit in the definition of PSCFs are independent. The PSCF that returns all maximal lotteries for the profile in which all preference relations are reversed violates unanimity but satisfies decisiveness, population-consistency, and composition-consistency.
When not requiring decisiveness, returning all \emph{ex post} efficient lotteries is consistent with the remaining axioms.\footnote{Continuity and convexity can also be seen as implicit assumptions. Continuity is needed because the relative interior of $\ml$ satisfies all remaining axioms. Whether convexity is required is open.}
\end{remark}

\begin{remark}[Size of Universe]
The proof of \thmref{thm:fisml} exploits the infinity of the universe. As stated in the previous remark, $\ml^3$ satisfies composition-consistency and violates population-consistency. However, $\ml^3$ does satisfy population-consistency when there are only up to three alternatives. This implies that the statement of \thmref{thm:fisml} requires a universe that contains at least four alternatives.
\end{remark}

\begin{remark}[Uniqueness]\label{rem:unique}
The set of profiles in which $\ml$ is not single-valued is negligible in the sense specified in the definition of PSCFs. 
When extending the set of fractional profiles to the reals, it can also be shown that maximal lotteries are almost always unique by using an argument similar to that of \citet{Hars73a}.
\end{remark}

\begin{remark}[Strong population-consistency]\label{rem:pcons}
$\ml$ does not satisfy the stronger version of population-consistency in which the set inclusion is replaced with equality (see \secref{sec:consistency}).\footnote{Strong population-consistency is quite demanding. It is for example violated by rather basic functions such as the Pareto rule.}  
This can be seen by observing that \emph{every} lottery is maximal for the union of any two electorates whose preferences are completely opposed to each other. When there are at least three alternatives, it is possible to find two such preference profiles which yield the same unique maximal lottery and strong population-consistency is violated.
However, whenever $\ml$ is single-valued (which is almost always the case), strong population-consistency is equivalent to population-consistency and therefore satisfied by $\ml$.
\end{remark}

\begin{remark}[Cloning-consistency and Condorcet-consistency]\label{rem:cloning}	
Requiring cloning-consistency instead of composition-consistency suffices for the proof of \thmref{thm:fisml} when additionally demanding Condorcet-consistency. It is therefore possible to alternatively characterize $\ml$ using population-consistency, cloning-consistency, and Condorcet-consistency. As above, the axioms are independent from each other.
$\ml^3$, as defined in \remref{rem:ind}, satisfies all axioms except population-consistency. The PSCF that is identical to $\ml^3$ for agendas of size~3 and otherwise identical to $\ml$ satisfies all axioms except cloning-consistency. $\rd$ satisfies all axioms except Condorcet-consistency.
\end{remark}

\begin{remark}[Agenda-consistency]
$\ml$ also satisfies \emph{agenda-consistency}, which requires that the set of all lotteries that are chosen from two overlapping agendas should be identical to the set of lotteries that are chosen from the union of both agendas (and whose support is contained in both agendas). The inclusion from left to right is known as Sen's $\gamma$ or \emph{expansion}, whereas the inclusion from right to left is Sen's $\alpha$ or \emph{contraction} \citep{Sen71a}.\footnote{Sen's $\alpha$ actually goes back to \citet{Cher54a} and \citet{Nash50b}, where it is called \emph{independence of irrelevant alternatives} (not to be confused with Arrow's IIA). We refer to \citet{Monj08a} for more details.}
Numerous impossibility results, including Arrow's well-known theorem, have revealed that agenda-consistency is prohibitive in non-probabilistic social choice when paired with minimal further assumptions such as non-dictatorship and Pareto-optimality \citep[\eg][]{Sen77a,Sen86a,CaKe02a}.\footnote{\citet{PaPe86a} obtained a similar impossibility for probabilistic social choice using an interpretation of Sen's $\alpha$ that is stronger than ours.}

\end{remark}

\begin{remark}[Domains]
In contrast to $\rd$, which at least requires that every voter has a unique top choice, $\ml$ does not require the asymmetry, completeness, or even transitivity of individual preferences (and still satisfies population-consistency and composition-consistency in these more general domains).
In the restricted domains of matching and house allocation, on the other hand, maximal lotteries are known as \emph{popular mixed matchings} \citep{KMN11a} or \emph{popular random assignments} \citep{ABS13a}.
\end{remark}

\begin{remark}[Efficiency]
It has already been observed by \citet{Fish84a} that $\ml$ is \emph{ex post} efficient, \ie Pareto-dominated alternatives always receive probability zero in all maximal lotteries. \citet{ABBH12a} strengthened this statement by showing that $\ml$ even satisfies $\mathit{SD}$-efficiency (also known as ordinal efficiency) as well as the even stronger notion of $\mathit{PC}$-efficiency \citep[see][]{ABB14b}. While $\rd$ also satisfies $\mathit{SD}$-efficiency, random serial dictatorship (the canonical generalization of $\rd$ to weak preferences) violates $\mathit{SD}$-efficiency \citep{BoMo01a,BMS05a}.
\end{remark}

\begin{remark}[Strategyproofness]
$\rd$ is the only \emph{ex post} efficient strategyproof PSCF \citep{Gibb77a}. However, $\rd$ and $\ml$ violate \emph{group}-strategyproofness (in fact, there exists no \emph{ex post} efficient group-strategyproof PSCF).
$\rd$ and $\ml$ satisfy the significantly weaker notion of $\mathit{ST}$-group-strategyproofness, which is violated by probabilistic extensions of most common voting rules \citep{ABB13d}.
A very useful property of $\ml$ is that it cannot be manipulated in all preference profiles that admit a strict Condorcet winner \citep[see][]{Peyr13a}.
\end{remark}

\section*{Acknowledgements}
This material is based on work supported by the Deutsche Forschungsgemeinschaft under grants {BR~2312/7-2} and {BR~2312/10-1}. Early versions of this paper were presented at the 6th Seminar on Ordered Structures in Games and Decisions in Paris (November, 2012), the Workshop on Mathematics of Electoral Systems in Budapest (October, 2013), the 12th Meeting of the Society for Social Choice and Welfare in Boston (June, 2014), the Micro-Theory Seminar at the University of Glasgow (February, 2015), and the Algorithms Seminar at the University of Oxford (November, 2015). The authors thank Bhaskar Dutta, Christian Geist, Paul Harrenstein, Johannes Hofbauer, Olivier Hudry, Vincent Merlin, Bernard Monjardet, Herv\'{e} Moulin, Klaus Nehring, Clemens Puppe, Peyton Young, and William S. Zwicker for helpful comments. 
The authors are furthermore indebted to the editor and anonymous referees for their constructive feedback.

\def\bibfont{\small}

\newpage

\setcounter{theorem}{0}

\appendix

\section*{APPENDIX}

\section{Preliminaries}

As stated in \secref{sec:prelims}, $U$ is an infinite set of alternatives.
For convenience we will assume that $\mathbb{N}\subseteq U$.
For $n\in\mathbb{N}$, $[n]$ is defined as $[n]=\{1,\dots,n\}$.
For two sets $A$ and $B$, let $\Pi(A,B)$ denote the set of all bijections from $A$ to $B$ (where $\Pi(A,B) = \emptyset$ if $|A|\neq |B|$).
Let $\Pi(A) = \Pi(A,A)$ be the set of all permutations of $A$.
We will frequently work with profiles in which alternatives are renamed according to some bijection from one set of alternatives to another. 
For all $A,B\in\fone(U)$, ${\pref}\in\mathcal{L}(A)$, and $\pi\in \Pi(A,B)$, let $\pi({\pref})=\{ (\pi(x),\pi(y)) \colon(x,y)\in {\pref}\}\in\mathcal{L}(B)$ and, for $R\in\mathcal{R}|_A$, let $\pi(R)\in\mathcal{R}|_B$ such that $R({{\pref}})= (\pi(R))(\pi({\pref}))$.
A well-known symmetry condition for PSCFs is \emph{neutrality}, which requires that all alternatives are treated equally in the sense that renaming alternatives is appropriately reflected in the outcome. Formally, a PSCF is neutral if
\[
\pi(f(R)) = f(\pi(R)) \text{ for all $A, B\in\fone(U)$, $R\in \mathcal{R}|_A$, and $\pi\in\Pi(A,B)$.} \tag{neutrality}
\]

We show that composition-consistency implies neutrality by replacing all alternatives with components of size~2.

\begin{lemma}\label{lem:neutrality}
	Every composition-consistent PSCF satisfies neutrality.
\end{lemma}

\begin{proof}
	Let $f$ be a composition-consistent PSCF, $A,B\in\fone(U)$, $R\in\mathcal{R}|_A$, and $\pi\in\Pi(A,B)$. 
	We have to show that $\pi(f(R)) = f(\pi(R))$. 
	To this end, let $p^A\in f(R)$.
	First, choose $A = \{a_1,\dots,a_n\}$ and $B = \{b_1,\dots, b_n\}$ such that $b_i = \pi(a_i)$ for all $i\in[n]$.
	Since $U$ is infinite, there is $C=\{c_1,\dots,c_n\}\in\fone(U)$ such that $C\cap A = \emptyset$ and $C\cap B = \emptyset$.
	Now, let $R'\in\mathcal{R}|_{A\cup C}$ such that $R'|_A = R$ and $\{a_i,c_i\}$ is a component in $R'$ for all $i\in[n]$.
	Thus, we have that $p^A\in f(R'|_A)$.
	We now apply composition-consistency to $a_i$ and the components $\{a_i,c_i\}$ for all $i\in[n]$, which by definition implies that
	\[
	f(R'|_A)\times_{a_1} f(R'|_{\{a_1,c_1\}})\times_{a_2} f(R'|_{\{a_2,c_2\}}) \dots\times_{a_n} f(R'|_{\{a_n,c_n\}}) = f(R'|_{A\cup C})\text.
	\]
	\todo{FFX: in the equation above and a couple of times below, it was wrongly $R$ instead of $R'$ in the published version.}
	Hence, for $p^{AC}\in f(R'|_{A\cup C})$ we have $p^{AC}_{a_i} + p^{AC}_{c_i} = p^A_{a_i}$ for all $i\in[n]$.
	Applying composition-consistency analogously to $c_i$ and $\{a_i,c_i\}$ for all $i\in[n]$ yields $p^C_{c_i} = p^{AC}_{a_i} + p^{AC}_{c_i} = p^A_{a_i}$ for all $p^C\in f(R'|_C)$ and $i\in[n]$.
Finally, let $R''\in\mathcal{R}|_{B\cup C}$ such that $R''|_{C} = R'|_{C}$ and $\{b_i,c_i\}$ is a component in $R''$ for all $i\in[n]$.
	Hence, we have that $p^C\in f(R''|_C)$.
	As before, it follows from composition-consistency that $p^B\in f(R''|_B)$ where $p^B_{b_i} = p^C_{c_i}$ for all $i\in[n]$.
	Notice that $p^B = \pi(p^A)$ and $B = \pi(A)$.
	Since $R''|_B = \pi(R)$ by construction of $R''$, we have $p^B\in f(\pi(R))$.
	Hence, $\pi(f(R))\subseteq f(\pi(R))$.
	The fact that $f(\pi(R))\subseteq\pi(f(R))$ follows from application of the above to $\pi(R)$ and $\pi^{-1}$.
\end{proof}

The following notation is required for our proofs.
For some set $X$, $\uni(X)$ denotes the \emph{uniform distribution} over $X$.
In particular, for $A\in\fone(U)$, $\uni(A)$ is the uniform lottery over $A$, \ie $\uni(A) = \nicefrac{1}{|A|}\sum_{x\in A}x$.
The \emph{support} of a lottery $p$ is the set of all alternatives to which $p$ assigns positive probability, \ie $\supp(p) = \{x\in A\colon p_x > 0\}$.
The \emph{$1$-norm} of $x\in\mathbb{Q}^n$ is denoted by $\|x\|$, \ie $\|x\| = \sum_{i=1}^n |x_i|$. 
For $X\subseteq\mathbb{Q}^n$, the \emph{convex hull} $\conv(X)$ is the set of all convex combinations of elements of $X$, \ie 
\[\conv(X)= \left\{\lambda_1 a^1 + \dots + \lambda_k a^k\colon a^i\in X, \lambda_i\in \mathbb{Q}_{\geq 0}, \sum_{i=1}^{k} \lambda_i = 1\right\}\text{.}\]
$X$ is \emph{convex} if $X = \conv(X)$.
The \emph{affine hull} $\aff(X)$ is the set of all affine combinations of elements of $X$, \ie 
\[\aff(X)= \left\{\lambda_1 a^1 + \dots + \lambda_k a^k\colon a^i\in X, \lambda_i\in \mathbb{Q}, \sum_{i=1}^{k} \lambda_i = 1\right\}\text{.}\]
$X$ is an \emph{affine subspace} if $X = \aff(X)$.
We say that $a^1,\dots,a^k\in\mathbb{Q}^n$ are \emph{affinely independent} if, for all $\lambda_1,\dots,\lambda_k \in\mathbb{Q}$ with $\sum_{i=1}^k \lambda_i = 0$, $\sum_{i=1}^k \lambda_i a^i = 0$ implies $\lambda_i=0$ for all $i\in[k]$.
The \emph{dimension} of an affine subspace $X$, $\dim(X)$, is $k-1$, where $k$ is the maximal number of affinely independent vectors in $X$.
The dimension of a set $X$ is the dimension of $\aff(X)$.
The \emph{linear hull} $\lin(X)$ is the set of all linear combinations of elements of $X$, \ie 
\[\lin(X)= \left\{\lambda_1 a^1 + \dots + \lambda_k a^k\colon a^i\in X, \lambda_i\in \mathbb{Q}\right\}\text{.}\]
$B_\epsilon(x)=\{y\in\mathbb{Q}^n\colon \|x-y\|<\epsilon\}$ denotes the \emph{$\epsilon$-ball} around $x\in\mathbb{Q}^n$.
The \emph{interior} of $X\subseteq\mathbb{Q}^n$ in $Y\subseteq\mathbb{Q}^n$ is $\int_Y(X) = \{x\in X\colon B_\epsilon(x)\cap Y\subseteq X \text{ for some } \epsilon>0\}$.
The \emph{closure} of $X\subseteq \mathbb{Q}^n$ in $Y\subseteq\mathbb{Q}^n$, $\cl_Y(X)$, is the set of all limit points of sequences in $X$ which converge in $Y$, \ie $\cl_Y(X) = \{\lim_{k\rightarrow \infty} a^k\colon (a^k)_{k\in\mathbb{N}} \text{ converges in } Y \text{ and $a^k\in X$ for all $k\in\mathbb{N}$}\}$.
$X$ is \emph{dense} in $Y$ if $\cl_Y(X) = Y$.
Alternatively, $X$ is dense at $y\in\mathbb{Q}^n$ if for every $\epsilon>0$ there is $x\in X$ such that $\|x-y\|<\epsilon$.
$X$ is dense in $Y$ if $X$ is dense at $y$ for every $y\in Y$.

\section{Non-Probabilistic Social Choice}

\begin{theorem}
	There is no SCF that satisfies population-consistency and composition-consistency.
\end{theorem}

\begin{proof}
	Assume for contradiction that $f$ is an SCF that satisfies population-consistency and composition-consistency. Let $A = \{a,b,c\}$ and consider the profiles $R^1,\dots,R^6$ as depicted below.
	
	\begin{center}
		\begin{tabular}{cccccc}
			\begin{tabular}{ccc}
			$\nicefrac{1}{3}$&$\nicefrac{1}{3}$&$\nicefrac{1}{3}$\\
			\midrule
			$a$&$b$&$c$\\
			$b$&$c$&$a$\\
			$c$&$a$&$b$\\[2ex]
			\multicolumn{3}{c}{$R^1$}
			\end{tabular}
			&
			\centering
			\begin{tabular}{cc}
			$\nicefrac{1}{2}$&$\nicefrac{1}{2}$\\
			\midrule
			$a$&$c$\\
			$b$&$b$\\
			$c$&$a$\\[2ex]
			\multicolumn{2}{c}{$R^2$}
			\end{tabular}       
			&
			\centering
			\begin{tabular}{cc}
			$\nicefrac{1}{2}$&$\nicefrac{1}{2}$\\
			\midrule
			$a$&$b$\\
			$c$&$c$\\
			$b$&$a$\\[2ex]
			\multicolumn{2}{c}{$R^3$}
			\end{tabular}        
			&
			\centering
			\begin{tabular}{cc}
			$\nicefrac{1}{2}$&$\nicefrac{1}{2}$\\
			\midrule
			$b$&$c$\\
			$a$&$a$\\
			$c$&$b$\\[2ex]
			\multicolumn{2}{c}{$R^4$}
			\end{tabular}        
			&
			\centering
			\begin{tabular}{cc}
			$\nicefrac{1}{2}$&$\nicefrac{1}{2}$\\
			\midrule
			$a$&$b$\\
			$b$&$c$\\
			$c$&$a$\\[2ex]
			\multicolumn{2}{c}{$R^5$}
			\end{tabular}        
			&
			\centering
			\begin{tabular}{cc}
			$\nicefrac{1}{2}$&$\nicefrac{1}{2}$\\
			\midrule
			$b$&$a$\\
			$a$&$c$\\
			$c$&$b$\\[2ex]
			\multicolumn{2}{c}{$R^6$}
			\end{tabular}
		\end{tabular}
	\end{center}
	
	We claim that $\Delta(\{a,b\})\subseteq f(R^i)$ for all $i\in\{1,\dots,6\}$. 
	It follows from neutrality that $f(R^1) = \Delta(A)$. 
	Again, by neutrality, $f(R^2|_{\{a,b\}}) = \Delta(\{a,b\})$ and $f(R^2|_{\{a,c\}}) = \Delta(\{a,c\})$.
	Notice that $\{a,b\}$ is a component in $R^2$.
	Hence, by composition-consistency,
	\[f(R^2) = f(R^2|_{\{a,c\}})\times_a f(R^2|_{\{a,b\}}) = \Delta(\{a,c\})\times_a \Delta(\{a,b\}) = \Delta(A)\text.\] 
	A similar argument yields $f(R^i) = \Delta(A)$ for $i = 3,4$. 
	Unanimity implies that $f(R^5|_{\{b,c\}}) = \{b\}$ and, by neutrality, we have $f(R^5|_{\{a,b\}}) = \Delta(\{a,b\})$. 
	Furthermore, $\{b,c\}$ is a component in $R^5$.
	Hence, by neutrality and composition-consistency,
	\[f(R^5) = f(R^5|_{\{a,b\}})\times_b f(R^5|_{\{b,c\}}) = \Delta(\{a,b\})\times_b \{b\} = \Delta(\{a,b\})\text.\] 
	Similarly, $f(R^6) = \Delta(\{a,b\})$.
	
	Every profile $R^i$ is a vector in the five-dimensional unit simplex $\mathcal{R}|_A$ in $\mathbb{Q}^6$. The corresponding vectors are depicted below.
	\[
	\left(
        \begin{array}{c}
			R^1\\
			R^2\\
			R^3\\
			R^4\\
			R^5\\
			R^6
     	\end{array}
	\right)
	=
		\left(
	        \begin{array}{cccccc}
				\nicefrac{1}{3} & \nicefrac{1}{3} & \nicefrac{1}{3} & 0 & 0 & 0\\
				\nicefrac{1}{2} & 0 & 0 & 0 & \nicefrac{1}{2} & 0\\
				0 & \nicefrac{1}{2} & 0 & \nicefrac{1}{2} & 0 & 0\\
				0 & 0 & \nicefrac{1}{2} & 0 & 0 & \nicefrac{1}{2}\\
				\nicefrac{1}{2} & \nicefrac{1}{2} & 0 & 0 & 0 & 0\\
				0 & 0 & 0 & \nicefrac{1}{2} & 0 & \nicefrac{1}{2}
	     	\end{array}
		\right)
	\]
It can be checked that $R^1,\dots,R^6$ are affinely independent, \ie $\dim(\{R^1,\dots,R^6\}) = 5$. It follows from population-consistency that $\Delta(\{a,b\})\subseteq f(R)$ for every $R\in\conv(\{R^1,\dots,R^6\})$. 
	Hence, $\{R\in\mathcal{R}|_A\colon |f(R)| = 1\}$ is not dense in $\mathcal{R}|_A$ at $\nicefrac{1}{6}\, R^1+\dots+\nicefrac{1}{6}\, R^6$, which contradicts decisiveness of $f$.
\end{proof}

\section{Probabilistic Social Choice}
\label{sec:fml}

In this section we prove that every PSCF that satisfies population-consistency and composition-consistency has to return maximal lotteries. The high-level structure of the proof is described after \thmref{thm:fisml} in \secref{sec:ml}.

\subsection{$\ml$ Satisfies Population-Consistency and Composition-Consistency}
\label{sec:mlsat}

We first show that $\ml$ is a PSCF that satisfies population-consistency and composition-consistency. This statement is split into two lemmas.

\begin{lemma}\label{lem:mlpscf}
$\ml$ is a PSCF.
\end{lemma}

\begin{proof}
	$\ml$ is continuous, since the correspondence that maps a matrix $M$ to the set of vectors $x$ such that $Mx \ge 0$ is (upper hemi-) continuous.

	The fact that $f(R)$ is convex for every $R\in\mathcal{R}$ follows from convexity of the set of maximin strategies for all (symmetric) zero-sum games.
	
	$\ml$ obviously satisfies unanimity by definition.
	
	$\ml$ satisfies decisiveness. Let $A\in\fone(U)$ and $R\in\mathcal{R}|_A$. 
	It is easy to see that, for every $\epsilon > 0$, we can find $R'\in B_\epsilon(R)\cap\mathcal{R}|_A$ and an odd integer $k$ such that $kR'(\pref)$ is an integer for every ${\pref}\in\mathcal L(A)$. 
	Then $kM_{R'}(x,y)$ is an odd integer for all distinct $x,y\in A$, since the numbers of voters who prefer $x$ to $y$ and $y$ to $x$ add up to the odd integer $k$.
	\citet{LLL97a} have shown that every symmetric zero-sum game whose off-diagonal entries are odd integers admits a unique Nash equilibrium. Applying their result to $kM_{R'}$, which has the same maximin strategies as $M_{R'}$, yields $|\ml(R')| = 1$ and hence $\ml$ is decisive. \end{proof}

Moreover, the set of symmetric zero-sum games with a unique maximin strategy inherits openness from the set of all zero-sum games with a unique maximin strategy \citep[][pp.~56--58]{BKS50a}. 
Hence, the set of profiles with a unique maximal lottery is open and dense in the set of all profiles and the set of profiles with multiple maximal lotteries is nowhere dense.

\begin{lemma}\label{lem:mlpc}
$\ml$ satisfies population-consistency and composition-consistency.
\end{lemma}

\begin{proof}
	To simplify notation, for every $v\in\mathbb{Q}^{n}$ and $X\subseteq [n]$, we denote by $v_{X}$ the restriction of $v$ to indices in $X$, \ie $v_{X} = (v_i)_{i\in X}$. 
	
	$\ml$ satisfies population-consistency. 
	Let $A\in\fone(U)$, $R', R''\in\mathcal{R}|_A$, and $p\in \ml(R')\cap\ml(R'')$. 
Then, by definition of $\ml$, $p^T M_{R'} q\ge 0$ and $p^T M_{R''}q\ge 0$ for all $q\in\Delta(A)$.
	Hence, for all $\lambda\in[0,1]$,
	\[
	p^T\left(\lambda M_{R'} + (1-\lambda) M_{R''}\right) q = \lambda\underbrace{p^T M_{R'} q}_{\ge 0} + (1-\lambda)\underbrace{p^T M_{R''} q}_{\ge 0}\ge 0\text,
	\]
	for all $q\in\Delta(A)$, which implies that $p\in\ml(\lambda R' + (1-\lambda)R'')$.

	$\ml$ satisfies composition-consistency. 
	Let $A',B\in\fone(U)$ such that $A'\cap B = \{b\}$, $A = A'\cup B$, and $R\in\mathcal{R}|_A$ such that $B$ is a component in $R$.
	To simplify notation, let $C = A \setminus B$ and $M = M_R, M_{A'} = M_{R|_{A'}}, M_B = M_{R|_B}$, and $M_C = M_{R|_{C}}$.
	Notice first that $M$ and $M_{A'}$ take the following form for some $v\in\mathbb{Q}^{A\setminus B}$:
		\[
		M=
			\left(
		        \begin{array}{ccc:ccc}
					\multicolumn{3}{c:}{\multirow{3}{*}{$M_C$}} & | & & |\\
					&&& v & \dots & v\\
					&&& | & & |\\
					\hdashline		
					\scalebox{2.0}[1.0]{\( - \)} & (-v^T) & \scalebox{2.0}[1.0]{\( - \)} & \multicolumn{3}{c}{\multirow{3}{*}{$M_B$}}\\
					& \vdots & \\
					\scalebox{2.0}[1.0]{\( - \)} & (-v^T) & \scalebox{2.0}[1.0]{\( - \)}\\
		       \end{array}
			\right)
			\text,\qquad\qquad
			M_{A'}=
				\left(
			        \begin{array}{ccc:c}
						\multicolumn{3}{c:}{\multirow{3}{*}{$M_C$}} & |\\
						&&& v\\
						&&& |\\
						\hdashline		
						\scalebox{2.0}[1.0]{\( - \)} & (-v^T) & \scalebox{2.0}[1.0]{\( - \)} & 0\\
			       \end{array}
				\right)\text{.}
		\]

	Let $p\in \ml(R|_{A'})\times_b \ml(R|_B)$. 
	Then, there are $p^{A'}\in \ml(R|_{A'})$ and $p^B\in \ml(R|_B)$ such that $p = p^{A'}\times_b p^B$.
	Let $q\in\Delta(A)$. Then,
	\begin{align*}
		p^T M q &= p_{C}^T M_{C} q_C + \|p_B\| (-v)^T q_C + p_C^Tv \|q_B\| + p_B^T M_B q_B\\
		&= (p_{C}, \|p_B\|)^T M_{A'} (q_{C}, \|q_B\|)^T + p_B^T M_B q_B\\
		&= \underbrace{(p^{A'})^T M_{A'} (q_{C}, \|q_B\|)^T}_{\ge 0} + \|p_B\|\underbrace{(p^B)^T M_B q_B}_{\ge 0} \geq 0\text,
	\end{align*}
		since $p^{A'}\in\ml(R|_{A'})$ and $p^B\in\ml(R|_B)$, respectively.
		Hence $p\in\ml(R)$.
	
	For the other direction, let $p\in \ml(R)$.
	We have to show that there are $p^{A'}\in\ml(R|_{A'})$ and $p^B\in\ml(R|_B)$ such that $p = p^{A'}\times_{b} p^B$.

	First, if $\|p_B\| = 0$ let $p^{A'} = p_{A'}$ and $p^B\in\ml(R|_B)$ be arbitrary.
	Let $q\in\Delta(A')$. Then,
	\begin{align*}
		(p^{A'})^T M_{A'} q = p_{C}^T M_{C} q_{C} + p_{C}^T v q_b =
		p^T M (q, 0)^T\ge 0\text,
	\end{align*}
	since $p\in\ml(R)$. Hence, $p^{A'}\in\ml(R|_{A'})$.
	
	Otherwise, let $p^{A'} = (p_C, \|p_B\|)$ and $p^B = p_B/\|p_B\|$.
	Let $q\in\Delta(A')$. Then,
	\begin{align*}
		(p^{A'})^T M_{A'} q &= p_{C}^T M_{C} q_{C} + \|p_B\| (-v)^T q_{C} + p_{C}^T v q_b\\
		&= p_{C}^T M_{C} q_{C} + \|p_B\| (-v)^T q_{C} + p_{C}^T v q_b + \frac{q_b}{\|p_B\|} \underbrace{p_B^T M_B p_B}_{=0}\\
		&= p^T M (q_{C}, \frac{q_b}{\|p_B\|} p_B)^T\ge 0\text.
	\end{align*}
	Hence, $p^{A'}\in\ml(R|_{A'})$.
	Let $q\in\Delta(B)$. 
Then,
	\begin{align*}
		\|p_B\|^2 (p^B)^T M_B q &= \|p_B\| p_B^T M_B q\\
		&= \|p_B\| p_B^T M_B q + \underbrace{p_{C}^T M_{C} p_{C}}_{=0} + \underbrace{\|p_B\| (-v)^T p_{C} + \|p_B\| p_{C}^T v}_{=0}\\
		&= (p_{C}, p_B)^T M (p_{C}, \|p_B\| q) = p^TM (p_{C}, \|p_B\| q) \geq 0\text.
	\end{align*}
	Hence, $p^B\in\ml(R|_B)$.
\end{proof}

\subsection{Binary Choice}\label{sec:binary}

The basis of our characterization of $\ml$ is the special case for agendas of size~2. 
The following lemma states that, on two alternatives, whenever a composition-consistent PSCF returns a non-degenerate lottery, it has to return all lotteries. Interestingly, the proof uses composition-consistency on three-element agendas, even though the statement itself only concerns agendas of size~2.
In order to simplify notation, define 
\[p^\lambda = \lambda a+(1-\lambda)b\text{.}\]

\begin{lemma}\label{lem:binary}
	Let $A=\{a,b\}$ and $f$ be a PSCF that satisfies composition-consistency. 
	Then, for all $R\in\mathcal{R}|_{A}$ and $\lambda\in(0,1)$, $p^{\lambda} \in f(R)$ implies $f(R) = \Delta(A)$.
\end{lemma}
\begin{proof}
	Let $R\in\mathcal{R}|_{A}$ and assume $p^{\lambda} \in f(R)$ for some $\lambda\in(0,1)$.
		Define $R'\in\mathcal{R}|_{\{a,b,c\}}$ as depicted below.
		
		\begin{center}
			\begin{tabular}{cc}
				$R(a,b)$&$R(b,a)$\\
				\midrule
				$a$&$c$\\
				$b$&$b$\\
				$c$&$a$\\[2ex]
				\multicolumn{2}{c}{$R'$}\\
			\end{tabular}        
		\end{center}	
	
	Notice that $R'|_{A} = R$ and thus, $p^{\lambda} \in f(R'|_A)$.
	Neutrality implies that $\lambda a+(1-\lambda )c \in f(R'|_{\{a,c\}})$.
	Since $A$ is a component in $R'$, we have $\lambda p^{\lambda} + (1-\lambda)c\in f(R'|_{\{a,c\}})\times_a f(R'|_A) = f(R')$.
	Since $\{b,c\}$ is also a component in $R'$, composition-consistency implies that $\lambda p^{\lambda} + (1-\lambda)c\in f(R') = f(R'|_{A})\times_b f(R'|_{\{b,c\}})$.
	Observe that $\lambda p^\lambda_a = p^{\lambda^2}_a$ and hence $p^{\lambda^2}\in f(R'|_A) = f(R)$.

	Applying this argument repeatedly, we get $p^{{\lambda^2}^k}\in f(R)$ for all $k\in \mathbb{N}$.
	Since ${\lambda^2}^k\rightarrow 0$ for $k\rightarrow\infty$ and $f$ is continuous, we get $p^0 = b\in f(R)$. 
	Similarly, it follows that $p^1 = a\in f(R)$. 
	The fact that $f$ is convex-valued implies that $f(R) = \Delta(A)$.
\end{proof}

The characterization of $\ml$ for agendas of size~2 proceeds along the following lines. By unanimity, neutrality, and \lemref{lem:binary}, we know which lotteries have to be returned by every composition-consistent PSCF for three particular profiles. Then population-consistency implies that every such PSCF has to return all maximal lotteries. Last, we again use population-consistency to show that the function is not decisive if it additionally returns lotteries that are not maximal. 

\begin{lemma}\label{lem:binary2}
	Let $f$ be a PSCF that satisfies population-consistency and composition-consistency and $A=\{a,b\}$. Then $f(R) = \ml(R)$ for every $R\in\mathcal{R}|_A$.
\end{lemma}

\begin{proof}
	First, note that $R\in\mathcal{R}|_A$ is fully determined by $R(a,b)$.
	Let $R\in\mathcal{R}|_A$ be the profile such that $R(a,b) = \nicefrac{1}{2}$.
	Since $f(R)\neq\emptyset$, there is $\lambda\in[0,1]$ such that $p^\lambda\in f(R)$.
	Neutrality implies that $p^{1-\lambda}\in f(R)$ and hence, by convexity of $f(R)$, $p^{\nicefrac{1}{2}} = \nicefrac{1}{2}\,(p^\lambda + p^{1-\lambda})\in f(R)$.
	If follows from \lemref{lem:binary} that $f(R) = \Delta(A)$.
	
	Now, let $R\in\mathcal{R}|_A$ be the profile such that $R(a,b) = 1$.
	Unanimity implies that $a\in f(R)$.
	By population-consistency and the first part of the proof, we get $a\in f(R')$ for all $R'\in\mathcal{R}|_A$ with $R'(a,b)\in[\nicefrac{1}{2},1]$.
	Similarly, $b\in f(R')$ for all $R'\in\mathcal{R}|_A$ with $R'(a,b)\in [0, \nicefrac{1}{2}]$.
	This already shows that $\ml(R)\subseteq f(R)$ for every $R\in\mathcal{R}|_A$.
	
	Finally, let $R\in\mathcal{R}|_A$ be a profile such that $R(a,b) = r > \nicefrac{1}{2}$.
	If $f(R)\neq\{a\}$, there is $\lambda\in[0,1)$ such that $p^\lambda\in f(R)$.
	We have shown before that $f(R') = \Delta(A)$ if $R'(a,b) = \nicefrac{1}{2}$.
	Hence, it follows from population-consistency that $p^\lambda\in f(R')$ for every $R'\in\mathcal{R}|_A$ with $R'(a,b)\in[\nicefrac{1}{2},r]$.
	But then $\{R\in\mathcal{R}|_A\colon R(a,b)\in[\nicefrac{1}{2},r]\}\subseteq \{R\in\mathcal{R}|_A\colon |f(R)|>1\}$ and hence, $\{R\in\mathcal{R}|_A\colon |f(R)|=1\}$ is not dense in $\mathcal{R}|_A$.
	This contradicts decisiveness of $f$. 
	An analogous argument shows that $f(R) = \{b\}$ whenever $R(a,b)<\nicefrac{1}{2}$.
	
	In summary, we have that $f(R) = \{a\}$ if $R(a,b)\in (\nicefrac{1}{2},1]$, $f(R) = \{b\}$ if $R(a,b)\in [0,\nicefrac{1}{2})$, and $f(R) = \Delta(A)$ if $R(a,b)=\nicefrac{1}{2}$. Thus, $f = \ml$ (as depicted in \figref{fig:twoalts}(a)).
\end{proof}

\subsection{$f\subseteq\ml$}

The first lemma in this section shows that every PSCF that satisfies population-consistency and composition-consistency is Condorcet-consistent for profiles that are close to the uniform profile $\uni(\mathcal{L}(A))$, \ie the profile in which every preference relation is assigned the same fraction of voters. We prove this statement by induction on the number of alternatives. Every profile close to the uniform profile that admits a Condorcet winner can be written as a convex combination of profiles that have a component and admit the same Condorcet winner. For these profiles we know from the induction hypothesis that the Condorcet winner has to be chosen.

\begin{lemma}\label{lem:pccond}
	Let $f$ be a PSCF that satisfies population-consistency and composition-consistency and $A\in\fone(U)$. Then, $f$ satisfies Condorcet-consistency in a neighborhood of the uniform profile $\uni(\mathcal{L}(A))$.
\end{lemma}

\begin{proof}
	Let $f$ be a PSCF that satisfies population-consistency and composition-consistency and $A\in\fone(U)$ with $|A| = n$. 
Let, furthermore, $R\in\mathcal{R}|_A$ be such that $a\in A$ is a Condorcet winner in $R$ and $\|R-\uni(\mathcal{L}(A))\| \le \epsilon_n = (4^n\Pi_{k = 1}^n k!)^{-1}$. We show that $a\in f(R)$ by induction over $n$.
	An example for $n = 3$ illustrating the idea is given after the proof.
For $n = 2$, the claim follows directly from \lemref{lem:binary2}. 
	
	For $n > 2$, fix $b\in A\setminus\{a\}$. First, we introduce some notation. For ${\pref}\in\mathcal{L}(A)$, we denote by ${\pref}^{-1}$ the preference relation that reverses all pairwise comparisons, \ie $x\pref^{-1} y$ iff $y\pref x$ for all $x,y\in A$. By $\pref^{b\rightarrow a}$ we denote the preference relation that is identical to $\pref$ except that $b$ is moved upwards or downwards (depending on whether $a\pref b$ or $b\pref a$) until it is next to $a$ in the preference relation. Formally, let
	\[
	X_\pref = 
	\begin{cases}
		\{x\in A\colon a\succ x\succ b\}\quad\text{if } a\pref b\text{, and}\\
		\{x\in A\colon b\succ x\succ a\}\quad\text{if } b\pref a\text,\\
	\end{cases}
	\]
	and
	\[
	{\pref^{b\rightarrow a}} = 
	\begin{cases}
		{\pref}\setminus(X_\pref\times\{b\})\cup (\{b\}\times X_\pref)\quad\text{if }a\pref b\text{, and}\\
		{\pref}\setminus(\{b\}\times X_\pref)\cup (X_\pref\times\{b\})\quad\text{if }b\pref a\text.
	\end{cases}
	\]
	Notice that for every ${\pref'}\in\mathcal{L}(A)$, there are at most $n-1$ distinct preference relations~$\pref$ such that ${\pref'} = {\pref^{b\rightarrow a}}$. Furthermore, we say that $\{a,b\}$ is a \emph{component in $\pref$} if $X_\pref = \emptyset$.
	
	We first show that composition-consistency implies Condorcet-consistency for a particular type of profiles. 
	For ${\pref}\in\mathcal{L}(A)$, let $S\in\mathcal{R}|_A$ such that $S({\pref}) + S({\pref}^{b\rightarrow a}) = S({\pref^{-1}}) = \nicefrac{1}{2}$. 
	We have that $S(a,x) = \nicefrac{1}{2}$ for all $x\in A\setminus \{a\}$ and hence, $a$ is a Condorcet winner in $S$. We prove that $a\in f(S)$ by induction over $n$. 
	For $n = 2$, this follows from \lemref{lem:binary2}. 
	For $n>2$, let $x\in A\setminus\{b\}$ such that $x\pref y$ for all $y\in A$ or $y\pref x$ for all $y\in A$.
	This is always possible since $n > 2$. 
	Notice that $A\setminus\{x\}$ is a component in $S$ and $S(x,y) = \nicefrac{1}{2}$ for all $y\in A\setminus\{x\}$. 
	If $x = a$, it follows from composition-consistency and \lemref{lem:binary2} that $a\in f(S)$. 
	If $x\neq a$, it follows from the induction hypothesis that $a\in f(S|_{A\setminus\{x\}})$.
	\lemref{lem:binary2} implies that $a\in f(S|_{\{a,x\}})$ as $S(a,x) = \nicefrac{1}{2}$.
	Then, it follows from composition-consistency that $a\in f(S|_{\{a,x\}}) \times_a f(S|_{A\setminus\{x\}}) = f(S)$.

	Now, for every ${\pref}\in\mathcal{L}(A)$ such that $\{a,b\}$ is not a component in $\pref$ and $0 < R({\pref}) \le R({\pref}^{-1})$, let $S^\pref\in\mathcal{R}|_A$ such that 
	\[S^\pref({\pref}) + S^\pref({\pref^{b\rightarrow a}}) = S^\pref({\pref^{-1}}) = \nicefrac{1}{2} \qquad\text{and}\qquad \nicefrac{S^\pref({\pref})}{S^\pref({\pref^{-1}})} = \nicefrac{R({\pref})}{R({\pref^{-1}})}\text.
	\]
From what we have shown before, it follows that $a\in f(S^\pref)$ for all ${\pref}\in\mathcal{L}(A)$.
	
	The rest of the proof proceeds as follows. We show that $R$ can be written as a convex combination of profiles of the type $S^\pref$ and a profile $R'$ in which $\{a,b\}$ is a component and $a$ is a Condorcet winner. 
Since $R$ is close to the uniform profile, $R({\pref'})$ is almost identical for all preference relations ${\pref'}$. Hence $S^\pref({\pref'})$ is close to $0$ for all preference relations $\pref'$ in which $\{a,b\}$ is a component.  
	As a consequence, $R'({\pref'})$ is almost identical for all preference relations $\pref'$ in which $\{a,b\}$ is a component and
	$R'|_{A\setminus\{b\}}$ is close to the uniform profile for $n-1$ alternatives, \ie $\uni(\mathcal{L}(A\setminus\{b\}))$. 
	By the induction hypothesis, $a\in f(R'|_{A\setminus\{b\}})$. 
	Since $\{a,b\}$ is a component in $R'$, it follows from composition-consistency that $a\in f(R')$.
	
We define
	\[
	S = 2 \sum_{\pref} R(\pref^{-1}) S^\pref\text,
	\]
	where the sum is taken over all $\pref$ such that $\{a,b\}$ is not a component in $\pref$ and $0 < R({\pref}) \le R({\pref^{-1}})$ (in case $R({\pref}) = R({\pref^{-1}})$ we pick one of $\pref$ and $\pref^{-1}$ arbitrarily).
	Now, let $R'\in\mathcal{R}|_A$ such that
\[R = (1 - \|S\|)R' + S\text.\]

	Note that, by definition of $S$, $R'(\pref) = 0$ for all ${\pref}\in\mathcal{L}(A)$ such that $\{a,b\}$ is not a component in $\pref$. Hence, $\{a,b\}$ is a component in $R'$.
	By the choice of $R$, we have that \[
	\|S\| = \sum_{\pref\in\mathcal{L}(A)} S(\pref) \le \frac{n! - 2(n-1)!}{n!} + \epsilon_n = 1- \frac{2}{n} + \epsilon_n\text.
	\]
	Using this fact, a simple calculation shows that
	\[
	R'(\pref) \le \frac{R(\pref) - S(\pref)}{\frac{2}{n} - \epsilon_n} \le \frac{\frac{1}{n!} + \epsilon_n}{\frac{2}{n} - \epsilon_n} \le \frac{1}{2(n-1)!} + \frac{\epsilon_{n-1}}{4(n-1)!}\text.
	\]
Since, for every preference relation $\pref$ where $\{a,b\}$ is a component, there is exactly one other preference relation identical to $\pref$ except that $a$ and $b$ are swapped, we have that
	\[
	R'|_{A\setminus\{b\}}(\pref)\le \frac{1}{(n-1)!} + \frac{\epsilon_{n-1}}{2(n-1)!}\text,
	\]
	for every ${\pref}\in\mathcal{L}(A\setminus\{b\})$.
	By the above calculation, we have that
	\[
	\left\|R'|_{A\setminus\{b\}} - \uni(\mathcal{L}(A\setminus\{b\}))\right\| \le \epsilon_{n-1}\text.
	\]
	Since $S^\pref(a,x) = \nicefrac{1}{2}$ for all $x\in A\setminus\{a\}$ and ${\pref}\in\mathcal{L}(A)$, we have that $R'(a,x)\ge\nicefrac{1}{2}$ for all $x\in A\setminus\{a\}$. 
	Thus, $a$ is a Condorcet winner in $R'|_{A\setminus\{b\}}$. From the induction hypothesis it follows that $a\in f(R'|_{A\setminus\{b\}})$. Using the fact that $R'(a,b)\ge\nicefrac{1}{2}$, \lemref{lem:binary2} implies that $a\in f(R'|_{\{a,b\}})$. 
	Finally, composition-consistency entails $a\in f(R'|_{A\setminus\{b\}}) \times_a f(R'|_{\{a,b\}}) = f(R')$.
	
	In summary, $a\in f(S^\pref)$ for all ${\pref}\in\mathcal{L}(A)$ and $a\in f(R')$. Since~$R$ is a convex combination of profiles of the type $S^\pref$ and~$R'$, it follows from population-consistency that~$a\in f(R)$.
	\end{proof}

	We now give an example for $A = \{a,b,c\}$ which illustrates the proof of \lemref{lem:pccond}. Consider the following profile $R$, where $0\le\epsilon\le \epsilon_3$.
\begin{center}
		\newlength{\w}
		\setlength{\w}{\widthof{$\nicefrac{(1+2\epsilon)}{6}$}}
	  	\begin{tabular}{P{\w} P{\w} P{\w} P{\w} P{\w} P{\w}}
			$\nicefrac{(1+2\epsilon)}{6}$ & $\nicefrac{1}{6}$ & $\nicefrac{1}{6}$ & $\nicefrac{(1-\epsilon)}{6}$ & $\nicefrac{(1-\epsilon)}{6}$ & $\nicefrac{1}{6}$\\
			\midrule
			$a$&$a$&$b$&$b$&$c$&$c$\\
			$b$&$c$&$a$&$c$&$a$&$b$\\
			$c$&$b$&$c$&$a$&$b$&$a$\\[2ex]
			\multicolumn{6}{c}{$R$}\\	 		
		\end{tabular}
	\end{center}
	Then, we have that $\|R - \uni(\mathcal{L}(A))\| \le \epsilon_3$. 
	Now consider ${\pref}$ with $b\succ c\succ a$, which yields $S^\pref$ as depicted below.	
\begin{center}
\setlength{\w}{\widthof{$\nicefrac{(1-\epsilon)}{2}$}}
	  	\begin{tabular}{P{\w} P{\w} P{\w}}
			$\nicefrac{1}{2}$ & $\nicefrac{(1-\epsilon)}{2}$ & $\nicefrac{\epsilon}{2}$\\
			\midrule
			$a$&$b$&$c$\\
			$c$&$c$&$b$\\
			$b$&$a$&$a$\\[2ex]
			\multicolumn{3}{c}{$S^\pref$}\\	 		
		\end{tabular}
	\end{center}
	Here, $y\pref a$ for all $y\in A$. Hence, it follows from what we have shown before that $a\in f(S^\pref)$. No other profiles of this type need to be considered, as $\pref$ and $\pref^{-1}$ are the only preference relations in which $\{a,b\}$ is not a component. Thus $S = \nicefrac{1}{3}\, S^\pref$.
	
	Then, we have $R'$, $R'|_{\{a,c\}}$, and $R'|_{\{a,b\}}$ as follows.
	\begin{center}
\setlength{\w}{\widthof{$\nicefrac{(1+2\epsilon)}{4}$}}
	  	\begin{tabular}{P{\w} P{\w} P{\w} P{\w}}
			$\nicefrac{(1+2\epsilon)}{4}$ & $\nicefrac{1}{4}$ & $\nicefrac{(1-\epsilon)}{4}$ & $\nicefrac{(1-\epsilon)}{4}$\\
			\midrule
			$a$&$b$&$c$&$c$\\
			$b$&$a$&$a$&$b$\\
			$c$&$c$&$b$&$a$\\[2ex]
			\multicolumn{4}{c}{$R'$}\\	 		
		\end{tabular}
		\qquad
	  	\begin{tabular}{cc}
			$\nicefrac{(1+\epsilon)}{2}$ & $\nicefrac{(1-\epsilon)}{2}$\\
			\midrule
			$a$&$c$\\
			$c$&$a$\\
			\\[2ex]
			\multicolumn{2}{c}{$R'|_{\{a,c\}}$}\\	 		
		\end{tabular}
		\qquad
	  	\begin{tabular}{cc}
			$\nicefrac{(2+\epsilon)}{4}$ & $\nicefrac{(2-\epsilon)}{4}$\\
			\midrule
			$a$&$b$\\
			$b$&$a$\\
			\\[2ex]
			\multicolumn{2}{c}{$R'|_{\{a,b\}}$}\\	 		
		\end{tabular}
	\end{center}
It follows from \lemref{lem:binary2} that $a\in f(R'|_{\{a,c\}})$ and $a\in f(R'|_{\{a,b\}})$. Then, composition-consistency implies that $a\in f(R') = f(R'|_{\{a,c\}})\times_a f(R'|_{\{a,b\}})$.
	
	In summary, we have that
	\[
	R = \nicefrac{2}{3}\, R' + \nicefrac{1}{3}\, S^\pref\text,
	\]
	$a\in f(R')$, and $a\in f(S^\pref)$. Thus, population-consistency implies that $a\in f(R)$.

\begin{lemma}\label{lem:cancel}
Every PSCF that satisfies population-consistency and composition-consistency returns the uniform lottery over all Condorcet winners for all profiles in a neighborhood of the uniform profile $\uni(\mathcal{L}(A))$. 
\end{lemma}

\begin{proof}
	Let $f$ be a PSCF that satisfies population-consistency and composition-consistency and $A\in\fone(U)$ with $|A| = n$.
	Moreover, let $R\in\mathcal{R}|_A$ such that $\|R-\uni(\mathcal{L}(A))\|\le\epsilon_n$ and $A'\subseteq A$ be the set of Condorcet winners in $R$.
We actually prove a stronger statement, namely that $\Delta(A')\subseteq f(R)$.
	Every alternative in $A'$ is a Condorcet winner in $R$.
	Thus, it follows from \lemref{lem:pccond} that $x\in f(R)$ for every $x\in A'$.
	Since $f(R)$ is convex, $\Delta(A')\subseteq f(R)$ follows.
\end{proof}

For the remainder of the proof, we need to define two classes of profiles that are based on regularity conditions imposed on the corresponding majority margins.
Let $A\in \mathcal{F}(U)$ and $A'\subseteq A$. A profile $R\in\mathcal{R}|_A$ is
\begin{align*}
	\text{\emph{regular on }} A' &\text{ if } \sum_{y\in A'} M_R(x,y) = 0 \text{ for all } x\in A'\text{, and}\\
	\text{\emph{strongly regular on }} A' &\text{ if } M_R(x,y) = 0 \text{ for all } x,y\in A'\text.
\end{align*}
By $\mathcal{R}|_A^{A'}$ and $\mathcal{S}|_A^{A'}$ we denote the set of all profiles in $\mathcal{R}|_A$ that are regular or strongly regular on $A'$, respectively.

In the following five lemmas we show that, for every $A'\subseteq A$, every profile on $A$ can be affinely decomposed into profiles of three different types: profiles that are strongly regular on $A'$, certain regular profiles, and profiles that admit a strict Condorcet winner in $A'$.\footnote{Similar decompositions of majority margin matrices have been explored by \citet{Zwic91a} and \citet{Saar95a}.} 
Lemmas \ref{lem:str-regular}, \ref{lem:skew-dim}, and \ref{lem:hamilton} do not make any reference to population-consistency, composition-consistency, or maximal lotteries and may be of independent interest.

First, we determine the dimension of the space of all profiles that are strongly regular on $A'$.

\begin{lemma}\label{lem:str-regular}
	Let $A'\subseteq A\in\fone(U)$. Then, $\dim(\mathcal{S}|_A^{A'})=|A|!-\binom{|A'|}{2} - 1$.
\end{lemma}

\begin{proof}
We will characterize $\mathcal{S}|_A^{A'}$ using a set of linear constraints.
By definition, $\mathcal{S}|_A^{A'} = \{R\in\mathcal{R}|_A\colon M_R(x,y) = 0 \text{ for all } x,y\in A'\}$.
	Recall that $M_R(x,y) = \sum_{\pref\colon x\pref y} R({\pref}) - \sum_{\pref\colon y\pref x} R({\pref})$.
	Since $M_R(x,x) = 0$ for all $R\in\mathcal{R}|_A$ and $x\in A$, $\mathcal{S}|_A^{A'}$ can be characterized by ${\phantom{'}|A|'\choose {2}}$ homogeneous linear constraints in the $(|A|!-1)$-dimensional space $\mathcal{R}|_A$, which implies that $\dim(\mathcal{S}|_A^{A'}) \geq |A|!-{\phantom{'}|A|'\choose {2}} - 1$. Equality holds but is not required for the following arguments. We therefore omit the proof.
\end{proof}

Second, we determine the dimension of the space of all skew-symmetric $n\times n$ matrices that correspond to regular profiles and vanish outside their upper left $n'\times n'$ sub-matrix, \ie
	\[
	\mathcal{M}_{n'} = \big\{M\in \mathbb{Q}^{n\times n}\colon M = -M^T\text{, } \sum_{j=1}^{n} M(i,j) = 0 \text{ if } i\in[n] \text{, and } M(i,j) = 0 \text{ if } \{i,j\}\not\subseteq [n']\big\} \text{.}
	\]

In \lemref{lem:hamilton}, we then proceed to show that every matrix of this type can be decomposed into matrices induced by a subset of regular profiles for which we know that every PSCF has to return the uniform lottery over the first $n'$ alternatives (possibly among other lotteries).

\begin{lemma}\label{lem:skew-dim}
$\dim(\mathcal{M}_{n'})={\phantom{'}n'\choose 2} - (n'-1)$.
\end{lemma}

\begin{proof}
	First note that the space of all $n\times n$ matrices has dimension $n^2$.
	We show that $\mathcal{M}_{n'}$ can be characterized by a set of $(n^2 - (n')^2) + ({\phantom{'}n'\choose 2} + n') + (n'-1)$ homogeneous linear constraints.
	Let $M\in\mathbb{Q}^{n\times n}$ and observe that $(n^2 - (n')^2)$ constraints are needed to ensure that $M$ vanishes outside of $[n']\times[n']$, ${\phantom{'}n'\choose 2} + n'$ constraints are needed to ensure skew-symmetry of $M$ on $[n']\times[n']$, \ie $M(i,j) = -M(j,i)$ for all $i,j\in[n']$, $j\ge i$, and $(n'-1)$ constraints are needed to ensure that the first $n'$ rows (and hence also the columns) of $M'$ sum up to $0$, \ie $\sum_{j=1}^{n} M(i,j) = 0$ for all $i\in [n'-1]$. 
	It follows from skew-symmetry and the latter $n'-1$ constraints that the $n'$th row of $M$ sums up to $0$, since
	\[\sum_{j = 1}^{n} M(n',j) = \sum_{i,j = 1}^{n} M(i,j) - \sum_{i = 1}^{n'-1}\sum_{j = 1}^{n} M(i,j) = 0\text.\] 
	The last $n - n'$ rows of $M$ trivially sum up to $0$.
	Hence, $\dim(\mathcal{M}_{n'})\ge (n')^2 -({\phantom{'}n'\choose 2} + n') - (n'-1) = {\phantom{'}n'\choose 2} - (n'-1)$. 
	Equality holds but is not required for the following arguments. We therefore omit the proof.
\end{proof}

Let $\Pi_{B}^\circ([n])$ be the set of all permutations that are cyclic on $B$ and coincide with the identity permutation outside of $B$.\footnote{A permutation $\pi\in\Pi([n])$ is cyclic on $B$ if $\pi^{|B|}$ is the smallest positive power of $\pi$ that is the identity function on $B$.} We denote by $\mathcal{M}_{n'}^\circ$ the space of all matrices in $\mathcal{M}_{n'}$ induced by a permutation in $\Pi_{B}^\circ([n])$ for some $B\subseteq[n']$, \ie
\[\mathcal{M}_{n'}^\circ = \big\{M\in \mathcal{M}_{n'}\colon M(i,j) = 
	\begin{cases}
		1 &\text{if } j = \pi(i)\text{, } i\in B\text,\\
		-1 &\text{if } i = \pi(j)\text{, } j\in B\text,\\
		0 &\text{otherwise,}
	\end{cases}
	\quad\text{for some } \pi\in\Pi_B^\circ([n]),\text{ } B\subseteq[n']\big\}\text,\]
	with the convention that $\mathcal{M}^\circ_{2} = \{0\}$. We now show that the linear hull of $\mathcal{M}_{n'}^\circ$ is $\mathcal{M}_{n'}$.

\begin{lemma}\label{lem:hamilton}
$\lin(\mathcal{M}_{n'}^\circ)=\mathcal{M}_{n'}$.
\end{lemma}

\begin{proof}
The idea underlying the proof is as follows: every matrix $M\in\mathcal{M}_{n'}$ corresponds to a weighted directed graph with vertex set $[n]$ where the weight of the edge from $i$ to $j$ is $M(i,j)$. If $M\neq 0$, there exists a cycle along edges with positive weight of length at least~3 in the subgraph induced by $[n']$. We obtain a matrix $M'$ with smaller norm than $M$ by subtracting the matrix in $\mathcal{M}_{n'}^\circ$ from $M$ that corresponds to the cycle identified before.
	
	Let $M\in\mathcal{M}_{n'}$ and $\kappa\in\mathbb{Q}_{>0}$ such that $\kappa M\in\mathbb{Z}^{n\times n}$.
We show, by induction over the nonnegative integer $\kappa\|M\|$, where $\|M\|=\sum_{i,j}|M(i,j)|$, that $M = \sum_{i = 1}^{\ell}\lambda_i M^i$ for some $\lambda_i\in\mathbb{Q}$ and $M^i\in\mathcal{M}_{n'}^\circ$ for all $i\in [\ell]$ for some $\ell\in\mathbb{N}$.
If $\kappa\|M\| = 0$ then $M = 0$. Hence, the induction hypothesis is trivial.
	
	If $\kappa\|M\|\neq 0$, \ie $M\neq 0$, we can find $B\subseteq[n']$ with $|B|\ge 3$ and $\pi\in\Pi_B^\circ([n])$ such that $M(i,j) > 0$ if $\pi(i) = j$ and $i\in B$. 
	Note that $\pi$ defines a cycle of length at least~3 in the graph that corresponds to $M$.
	We define $M^1\in\mathcal{M}_{n'}^\circ$ by letting
	\[M^1(i,j) = 
	\begin{cases}
		1 &\text{if } \pi(i) = j\text{ and } i\in B\text,\\
		-1 &\text{if } \pi(j) = i\text{ and } j\in B\text{, and}\\
		0 &\text{otherwise.}
	\end{cases}
	\]
	Let $\lambda = \min\{M(i,j)\colon i,j\in[n] \text{ and } M^1(i,j) > 0\}$ and $M' = M - \lambda M^1$.
	By construction, we have that $M'(i,j) = M(i,j) - \lambda$ if $\pi(i) = j$ and $i\in B$, $M'(i,j) = M(i,j) + \lambda$ if $\pi(j) = i$ and $j\in B$, and $M'(i,j) = M(i,j)$ otherwise.
	Note that $M(i,j)\ge\lambda$ if $\pi(i) = j$ and $i\in B$ and $M(i,j)\le-\lambda$ if $\pi(j) = i$ and $j\in B$ by definition of $\lambda$.
	Recall that $\kappa M\in\mathbb{Z}^{n\times n}$ and, in particular, $\kappa\lambda\in\mathbb{N}$.
	Hence, $\kappa M'\in\mathbb{Z}^{n\times n}$.
	Moreover, $\kappa\|M'\| = \kappa\|M\| - 2\kappa \lambda |B|\le \kappa\|M\| - 1$.
	From the induction hypothesis we know that $M' = \sum_{i = 2}^{\ell} \lambda_i M^i$ with $\lambda_i\in\mathbb{Q}$ and $M^i\in\mathcal{M}_{n'}^\circ$ for all $i\in[\ell]$ for some $\ell\in\mathbb{N}$.
	By construction of $M'$, we have that $M = \sum_{i = 1}^{\ell} \lambda_i M^i$ with $\lambda_1 = \lambda$.
\end{proof}

\lemref{lem:regular} leverages Lemmas~\ref{lem:cancel},~\ref{lem:str-regular},~\ref{lem:skew-dim}, and~\ref{lem:hamilton} to show two statements. First, it identifies the dimension of the space of all profiles that are regular on~$A'\subseteq A$. Second, it proves that there is a full-dimensional subset of the space of all profiles that are regular on~$A'$ for which every PSCF that satisfies population-consistency and composition-consistency returns the uniform lottery over $A'$.

\begin{lemma}\label{lem:regular}
Let $f$ be a PSCF that satisfies population-consistency and composition-consistency and $A'\subseteq A\in\fone(U)$. Then, there is $\mathcal{X}\subseteq\mathcal{R}|_A^{A'}$ of dimension $|A|!-|A'|$ such that $\uni(A')\in f(R)$ for every $R\in \mathcal{X}$.
\end{lemma}

\begin{proof}
	To simplify notation, we assume without loss of generality that $A = [n]$ and $A'=[n']$. For $M\in \mathbb{Q}^{n\times n}$ and $\pi\in\Pi(A)$, let $\pi(M)$ be the matrix that results from $M$ by permuting the rows and columns of $M$ according to $\pi$, \ie $(\pi(M))(i,j) = M(\pi(i),\pi(j))$.
	
	From \lemref{lem:str-regular} we know that we can find a set $\mathcal{S} =\{S^1,\dots,S^{n!-{\phantom{'}n'\choose 2}}\}\subseteq\mathcal{R}|_{[n]}^{[n']}$ of affinely independent profiles.
	Since $\mathcal{S}$ can be chosen such that every $S\in\mathcal{S}$ is close to $\uni(\mathcal{L}([n]))$, it follows from \lemref{lem:cancel} that $\uni([n'])\in f(S)$ for all $S\in\mathcal{S}$.
	Therefore, it suffices to find a set of profiles $\mathcal{T}=\{R^1,\dots,R^{{\phantom{'}n'\choose 2}-(n'-1)}\}\subseteq\mathcal{R}|_{[n]}^{[n']}$ such that $\uni([n'])\in f(R)$ for every $R\in\mathcal{T}$ and $\mathcal{S}\cup\mathcal{T}$ is a set of affinely independent profiles.
	If $n' = 2$, we can choose $\mathcal{T} = \emptyset$.
	For $n' \ge 3$ we construct a suitable set of profiles as follows.

	For every $B\subseteq [n']$ with $|B| = k\ge 3$ and $\pi\in\Pi_B^\circ([n])$, let $[n]\setminus B = \{a_1,\dots, a_{n-k}\}$ and $R_B^\pi$ be defined as follows: $R_B^\pi(\pref) = \nicefrac{1}{(2k)}$ if
	\begin{align*}
		&\pi^0(i)\pref\pi^1(i)\pref\pi^2(i)\pref\dots\pref\pi^{k-1}(i)\pref a_1\pref\dots\pref a_{n-k} \quad\text{or}\\
		&a_{n-k}\pref\dots\pref a_1\pref \pi^{k-1}(i)\pref\dots\pref\pi^2(i)\pref\pi^0(i)\pref\pi^1(i)\text,
	\end{align*}
	for some $i\in B$.
	Note that $R_B^\pi$ is regular on $[n']$, since
	\[
	M_{R_B^\pi}(i,j) = 
	\begin{cases}
	\lambda \quad&\text{if } \pi(i) = j \text{ and } i\in B\text,\\
	-\lambda \quad&\text{if } \pi(j) = i \text{ and } j\in B\text{, and}\\
	0 \quad&\text{otherwise,}
	\end{cases}
	\]
	where $\lambda = \nicefrac{1}{k} > 0$.
	Hence, for every $M\in\mathcal{M}_{n'}^\circ$, there are $B\subseteq[n']$ and $\pi\in\Pi_B^\circ([n])$ such that $\lambda M = M_{R_B^\pi}$.
	Notice that $B$ and $[n]\setminus B$ are components in $R_B^\pi$.
	For $j\in B$, we have by construction that $R_B^\pi(j,a_1) = 0$.
	Hence, it follows from \lemref{lem:binary2} that $j\in f(R_B^\pi|_{\{j,a_1\}})$ and $a_1\in f(R_B^\pi|_{\{j,a_1\}})$.
	Moreover, neutrality, convexity, and composition-consistency imply that $\uni(B)\in f(R_B^\pi)$ by the symmetry of $R_B^\pi$ with respect to $B$.
	Now let $a_i\in \{a_1,\dots,a_{n-k}\}$.
	Observe that $\{a_1,\dots,a_{i-1}\}$ is a component in $R_B^\pi$ and $R_B^\pi(a_1,a_{i}) = 0$.
	Thus, composition-consistency and \lemref{lem:binary2} imply that 
	\[
	a_i\in f(R_B^\pi|_{\{a_1,a_i\}})\times_{a_1} f(R_B^\pi|_{\{a_1,\dots,a_{i-1}\}}) = f(R_B^\pi|_{\{1,\dots,i\}})\text.
	\]
	Furthermore, $\{a_{i+1},\dots,a_{n-k}\}$ is a component in $R_B^\pi$ and $R_B^\pi(a_{i},a_{n-k}) = 0$.
	As before, we get
	\[
	a_i\in f(R_B^\pi|_{\{a_i,a_{n-k}\}})\times_{a_{n-k}} f(R_B^\pi|_{\{a_{i+1},\dots,a_{n-k}\}}) = f(R_B^\pi|_{\{i,\dots,n-k\}})\text.
	\]
	Also $\{a_i,\dots,a_{n-k}\}$ is a component in $R_B^\pi$ and thus,
	\[
	a_i\in f(R_B^\pi|_{\{a_1,\dots,a_i\}})\times_{a_i} f(R_B^\pi|_{\{a_i,\dots,a_{n-k}\}}) = f(R_B^\pi|_{[n]\setminus B})\text.
	\]
	As $B$ is a component in $R_B^\pi$ and $R_B^\pi(j,a_1) = 0$, we get 
	\[
	a_i\in f(R_B^\pi|_{\{j,a_1,\dots,a_i\}})\times_{j} f(R_B^\pi|_{B}) = f(R_B^\pi)\text.
	\]
	Then, it follows from convexity of $f(R_B^\pi)$ that
	\[
	\uni([n']) = \frac{k}{n'}\uni(B) + \frac{1}{n'}\sum_{a_i\in [n']\setminus B} a_i\in f(R_B^\pi)\text,
	\]
	since $\uni(B)\in f(R_B^\pi)$ and $a_i\in f(R_B^\pi)$ for every $i\in [n-k]$.

	We know from \lemref{lem:hamilton} that $\lin(\mathcal{M}_{n'}^\circ) = \mathcal{M}_{n'}$ and, by \lemref{lem:skew-dim}, $\dim(\mathcal{M}_{n'}^\circ) \ge {\phantom{'}n'\choose 2} - (n'-1)$.
	Thus, we can find a basis $\{M^1,\dots, M^{{\phantom{'}n' \choose 2} - (n'-1)}\}$ of $\mathcal M_{n'}^\circ$ and a set of corresponding profiles
	\[
	\mathcal{T} = \{R^1,\dots, R^{{\phantom{'}n' \choose 2} - (n'-1)} \}\subseteq\{R_B^\pi\colon B\subseteq [n']\text{ and } \pi\in\Pi_B^\circ\}\text.
	\]
	We claim that $\mathcal{S}\cup \mathcal{T}$ is a set of affinely independent profiles. 
	Let $S^1,\dots,S^l\in \mathcal{S}$ and $R^1,\dots,R^m\in\mathcal{T}$ be pairwise disjoint. 
	Assume that $\sum_i \lambda_i S^i + \sum_j \mu_j R^j = 0$ for some $\lambda^i,\mu^j\in\mathbb{Q}$ such that $\sum_i\lambda_i + \sum_j\mu_j = 0$. 
	This implies that $\sum_j \mu_j M^j = 0$, which in turn implies $\mu_j=0$ for all $j\in[m]$, since the $M^j$'s are linearly independent. 
	Hence, $\sum_i\lambda_i S^i = 0$ and $\sum_i\lambda_i = 0$, which implies that $\lambda_i = 0$ for all $i\in[l]$, since $S^1,\dots,S^{n!-{\phantom{'}n'\choose 2}}$ are affinely independent. 
	Thus, $\mathcal{S}\cup\mathcal{T}$ is a set of affinely independent profiles and $\dim(\mathcal{S}\cup\mathcal{T}) = |\mathcal{S}\cup\mathcal{T}|-1 = n!-n'$. 
	The above stated fact that $\uni([n'])\in f(R_B^\pi)$ for every $B\subseteq [n']$ and $\pi\in\Pi_B^\circ([n'])$ finishes the proof.
\end{proof}

We now consider PSCFs that may return a lottery that is \emph{not} maximal. The following lemma shows that for every such PSCF there is a set of profiles with a strict Condorcet winner for which it returns the uniform lottery over a fixed subset of alternatives if we additionally require population-consistency and composition-consistency. Furthermore, this set of profiles has only one regular profile in its linear hull. Later this statement is leveraged to show that every 
population-consistent and composition-consistent PSCF returns a subset of maximal lotteries.

\begin{lemma}\label{lem:cond-dim}
Let $f$ be a PSCF that satisfies population-consistency and composition-consistency. 
If $f\not\subseteq \ml$, there are $A'\subseteq A\in\fone(U)$ with $|A'|\ge 2$ such that, for every $\epsilon>0$, there is $\mathcal{Y}\subseteq \mathcal{R}|_A$ of dimension $|A'|-1$ whose profiles are within distance $\epsilon$ of $\uni(\mathcal L(A))$ and have a strict Condorcet winner in $A'$, such that $\uni(A')\in f(R)$ for every $R\in \mathcal{Y}$ and $\dim(\lin(\mathcal{Y})\cap\lin(\mathcal{R}|_A^{A'})) = 1$.
\end{lemma}

\begin{proof}
	If $f\not\subseteq \ml$, there are $A\in\fone(U)$, $R\in\mathcal{R}|_A$, and $p\in f(R)$ such that $p\not\in \ml(R)$. 
	Since $p$ is not a maximal lottery, by definition, there is $q\in \Delta(A)$ such that $q^TM_Rp > 0$.
	Linearity of matrix multiplication implies that there is $x\in A$ such that $(M_Rp)_x > 0$, where $(M_Rp)_x$ is the entry of $M_Rp$ corresponding to $x$.
	We first use composition-consistency to ``blow up'' alternatives such that the resulting lottery is uniform.
	Let $\gcd$ be the greatest common divisor of $\{p_y\colon y\in A\}$, \ie $\gcd = \max\{s\in\mathbb{Q}\colon \nicefrac{p_y}{s}\in\mathbb{N} \text{ for all } y\in A\}$.
	For every $y\in A$, let $A_y\in\fone(U)$ such that $|A_y| = \max\{1,\nicefrac{p_y}{\kappa}\}$, $A_y\cap A = \{y\}$, and all $A_y$ are pairwise disjoint.
	Moreover, let $A^u = \bigcup_{y\in A} A_y$.
	Now, choose $R^u\in\mathcal{R}|_{A^u}$ such that $R^u|_A = R$, $A_y$ is a component in $R^u$ for every $y\in A$, and $R^u|_{A_y} = \uni(\mathcal{L}(A_y))$ for every $y\in A_y$.	
	Hence, $\uni(A_y)\in f(R^u|_{A_y})$ for all $y\in A$ as $f$ is neutral and $f(R^u|_{A_y})$ is convex.
	To simplify notation, let $A^p = \bigcup_{y\in\supp(p)} A_y$.
	By composition-consistency,
it follows that $p' = \uni(A^p)\in f(R^u)$.
	Observe that
	\[
	(M_{R^u}p')_x = \sum_{y\in \supp(p)\setminus\{x\}} \frac{|A_y|}{|A^p|}M_{R^u}(x,y) = \sum_{y\in A\setminus\{x\}} p_yM_R(x,y) > 0\text.
	\] 
	
	We now construct a profile $R'\in\mathcal{R}|_{A^u}$ such that $x$ is a strict Condorcet winner in $R'$ and $\uni(A^p) \in f(R')$. 
	To this end, let $R'\in\mathcal{R}|_{A^u}$ be the uniform mixture of all profiles that arise from $R^u$ by permuting all alternatives in $A^p\setminus\{x\}$, \ie
	\[
	R' = \frac{1}{|A^p\setminus\{x\}|!} \sum_{\substack{\pi\in\Pi(A^u)\colon\pi(y) = y\\\text{for all } y\in A^u\setminus A^p \cup \{x\}}} \pi(R^u)\text.
	\] 
	Then, $M_{R'}(x, y) = M_{R'}(x, z) > 0$ for all $y,z\in A^p\setminus\{x\}$. 
	Neutrality and population-consistency imply that $p'\in f(R')$.
	
	Let $R^\uni = \uni(\mathcal{L}(A^u))$ and define, for $\lambda\in [0,1]$, 
	\[
	R^\lambda = \lambda R' + (1-\lambda) R^\uni\text.
	\]
	It follows from \lemref{lem:pccond} that $y\in f(R^\uni)$ for all $y\in A^u$. Convexity of $f(R^\uni)$ implies that $f(R^\uni) = \Delta(A^u)$. Hence, by population-consistency, $p'\in f(R^\lambda)$ for all $\lambda\in [0,1]$.
	
	Now, let $S\in\mathcal{R}|_{A^u}$ such that $M_S(y,z) = 0$ for all $y,z\in A^p\cup\{x\}$ and $M_S(y,z) = 1$ for all $y\in A^p\cup\{x\}$, $z\in A^u\setminus (A^p\cup\{x\})$. 
	For $\lambda\in [0,1]$, let
	\[S^\lambda = \lambda S + (1-\lambda) R^\uni\text.\]
	Note that every $y\in A^p\cup\{x\}$ is a Condorcet winner in $S^\lambda$.
	It follows from population-consistency and \lemref{lem:pccond} that, for small $\lambda > 0$, $y\in f(S^\lambda)$ for all $y\in A^p\cup\{x\}$ and, by convexity, $\Delta(A^p\cup\{x\})\subseteq f(S^\lambda)$.
	In particular, $p'\in f(S^\lambda)$ for small $\lambda > 0$.
	
	Finally, let
	\[R^x = \nicefrac{1}{3}\, R^\lambda + \nicefrac{2}{3}\, S^\lambda\text,\]
	for some small $\lambda > 0$. Population-consistency implies that $p'\in f(R^x)$.
	Moreover, $M_{R^x}(x,y) > 0$ for all $y\in A^u\setminus\{x\}$, \ie $x$ is a strict Condorcet winner in $R^x$, and hence, it follows from \lemref{lem:pccond} that $x\in f(R^x)$.
	
	If $p_x > 0$ then, by construction, $p' = \uni(A^p\cup\{x\})\in f(R^x)$.
	If $p_x = 0$ then $p' = \uni(A^p)\in f(R^x)$.
	In this case it follows from convexity of $f(R^x)$ that $\uni(A^p\cup\{x\}) = \nicefrac{1}{(|A^p|+1)}\, x + \nicefrac{|A^p|}{(|A^p|+1)}\, \uni(A^p)\in f(R^x)$.
	
	Hence, in either case, we get a profile $R^x$ such that $\uni(A^p\cup\{x\})\in f(R^x)$ and $M^x = M_{R^x}$ restricted to $A^p\cup\{x\}$ takes the form
	\begin{align*}
		\begin{array}{cccc}
				M^x &=& \lambda& \cdot
		\end{array}
		&
		\left(
		\begin{array}{c c c c c c c}
		0 & \dots & 0 & -1 & 0 & \dots & 0 \\
		\vdots & \ddots & \vdots & \vdots & \vdots & \ddots & \vdots \\
		0 & \dots & 0 & -1 & 0 & \dots & 0 \\
		1 & \dots & 1 & 0 & 1 & \dots & 1 \\
		0 & \dots & 0 & -1 & 0 & \dots & 0 \\
		\vdots & \ddots & \vdots & \vdots & \vdots & \ddots & \vdots \\
		0 & \dots & 0 &-1 & 0 & \dots & 0 \\
		\end{array}
		\right)
	\end{align*}
	for some $\lambda > 0$ where all entries except the $x$th row and column are zero.
	Let $n' = |A^p\cup\{x\}|$.
	For every $y\in A^p\cup\{x\}$, let $M^y\in\mathbb{Q}^{n'\times n'}$ such that $M^y(y,z) = - M^y(z,y) = \lambda$ for all $z\neq y$ and $0$ otherwise.
	Let $\pi^y\in\Pi(A^u)$ such that $\pi^y(x)=y$ and $\pi^y(z) = z$ for all $z\in A^u\setminus (A^p\cup\{x\})$ and $R^y = \pi^y(R^x)$. 
	Then, for every $y\in A^p\cup\{x\}$, the restriction of $M_{R^y}$ to $A^p\cup\{x\}$ is $M^y$ and, by neutrality, $\uni(A^p\cup\{x\})\in f(R^y)$.
	
	Let $\mathcal{Y} = \{R^y\colon y\in A^p\cup\{x\}\}$. 
The profiles in $\mathcal Y$ are affinely independent. Indeed, if $\sum_y\alpha_yR^y=0$ and $\sum_y\alpha_y=0$, then the row sums of $\sum_y\alpha_yM^y=0$ imply that $n'\alpha_y=0$ for every $y$. Thus, $\dim(\mathcal{Y}) = |A^p\cup\{x\}|-1$.
	Now we determine  $\dim(\lin(\mathcal{Y})\cap\lin(\mathcal{R}|_{A^u}^{A^p\cup\{x\}}))$.
	To this end, let $Z=\sum_{z\in A^p\cup\{x\}}\alpha_z R^z\in\lin(\mathcal{Y})\cap\lin(\mathcal{R}|_{A^u}^{A^p\cup\{x\}})$.
	By linearity of majority margins, the restriction of the majority margin matrix of $Z$ to $A^p\cup\{x\}$ is $\sum_{z\in A^p\cup\{x\}}\alpha_z M^z$.
	Since $Z$ lies in the linear hull of profiles that are regular on $A^p\cup\{x\}$, the row sums of this matrix vanish. Thus,
	\[
	(n'-1)\alpha_y=\sum_{z\in A^p\cup\{x\}\setminus\{y\}}\alpha_z
	\]
	for all $y\in A^p\cup\{x\}$. Hence, all coefficients $\alpha_y$ are equal.
	Conversely, $\nicefrac{1}{n'}\sum_{y\in A^p\cup\{x\}}R^y$ is a regular profile, so
	\[
	\lin(\mathcal{Y})\cap\lin(\mathcal{R}|_{A^u}^{A^p\cup\{x\}}) = \left\{\alpha \sum_{y\in A^p\cup\{x\}}R^y\colon \alpha\in\mathbb{Q}\right\}.
	\]
	The parameter $\lambda$ in the construction of $R^x$ can be chosen arbitrarily small. Since every $R^y$ is obtained from $R^x$ by renaming alternatives, $\mathcal{Y}$ can therefore be chosen in any neighborhood of $\uni(\mathcal L(A^u))$, and every $R^y$ has $y$ as a strict Condorcet winner.
	Finally, $|A^p\cup\{x\}|\ge2$: this is immediate if $x\notin\supp(p)$, while $p=e_x$ would contradict $(M_Rp)_x>0$ if $x\in\supp(p)$ and $A^p\cup\{x\}$ were a singleton.
	Choosing $A=A^u$ and $A'=A^p\cup\{x\}$ proves the claim.
\end{proof}

	In \lemref{lem:fml}, we finally show that every PSCF that satisfies population-consistency and composition-consistency has to yield maximal lotteries. The structure of the proof is as follows. We assume for contradiction that a PSCF satisfies population-consistency and composition-consistency, but returns a lottery that is not maximal. Then we can find a full-dimensional convex set of profiles for which the uniform lottery over a fixed subset of at least two alternatives is returned. This set contains a profile in its interior that is close to the uniform profile and has a strict Condorcet winner. For every profile in an $\epsilon$-ball around this profile, the function has to return the uniform lottery over a non-singleton subset as well as the lottery with probability~1 on the Condorcet winner, which contradicts decisiveness.

\begin{lemma}\label{lem:fml}
Every PSCF $f$ that satisfies population-consistency and composition-consistency has to yield maximal lotteries, \ie $f\subseteq \ml$.
\end{lemma}

\begin{proof}
	Let $f$ be a PSCF that satisfies population-consistency and composition-consistency.
	For agendas of size $2$, the statement follows from \lemref{lem:binary2}. Assume for contradiction that $f\not\subseteq \ml$. 
	By \lemref{lem:cond-dim}, there are $A'\subseteq A\in\fone(U)$ with $|A'|\ge2$ such that the conclusion of that lemma holds for every $\epsilon>0$.
	By \lemref{lem:pccond}, there is a neighborhood $N$ of $\uni(\mathcal L(A))$ in which $f$ is Condorcet-consistent. Choose $\delta>0$ such that $B_\delta(\uni(\mathcal L(A)))\cap\mathcal R|_A\subseteq N$, and choose the set $\mathcal Y$ from \lemref{lem:cond-dim} within $B_{\delta/2}(\uni(\mathcal L(A)))$.
	By \lemref{lem:regular}, there is $\mathcal{X}\subseteq\mathcal{R}|_A^{A'}$ of dimension $|A|!-|A'|$ such that $\uni(A')\in f(R)$ for every $R\in\mathcal{X}$.
	Since every profile has coordinate sum~$1$, neither $\aff(\mathcal X)$ nor $\aff(\mathcal Y)$ contains~$0$. Hence, $\lin(\mathcal{X})$ has dimension $|A|! - |A'| + 1$ and $\lin(\mathcal{Y})$ has dimension $|A'|$.
	Moreover, $\lin(\mathcal X)\subseteq\lin(\mathcal R|_A^{A'})$, so \lemref{lem:cond-dim} implies that $\lin(\mathcal X)\cap\lin(\mathcal Y)$ has dimension at most~$1$. The dimension formula gives the reverse inequality because the two linear hulls have dimensions summing to $|A|!+1$ in $\mathbb Q^{|A|!}$. Their intersection therefore has dimension~$1$, and $\lin(\mathcal{X}\cup\mathcal{Y})$ has dimension $|A|!$.
	This implies that $\mathcal{X}\cup\mathcal{Y}$ has dimension $|A|!-1$.

	Furthermore, it follows from population-consistency that $\uni(A')\in f(R)$ for every $R\in\conv(\mathcal{X}\cup \mathcal{Y})$. 
	Let $C=\conv(\mathcal{X}\cup\mathcal{Y})$. Since $C$ is full-dimensional in $\mathcal R|_A$, choose $R^I\in\int_{\mathcal R|_A}(C)$.
	Let $R^0\in\mathcal Y$ and let $x\in A'$ be a strict Condorcet winner in $R^0$. For every sufficiently small positive rational $\eta$, the profile $R^x=(1-\eta)R^0+\eta R^I$ lies in $\int_{\mathcal R|_A}(C)$ and $B_\delta(\uni(\mathcal L(A)))$, and still has $x$ as a strict Condorcet winner.
Hence, there is $\epsilon>0$ such that, for every $R'\in B_\epsilon(R^x)\cap\mathcal{R}|_A$, $R'\in C\cap B_\delta(\uni(\mathcal L(A)))\subseteq C\cap N$ and $x$ is a strict Condorcet winner in $R'$. 
	Then, we get that $x\in f(R')$ and $\uni(A')\in f(R')$ for every $R'\in B_\epsilon(R^x)\cap\mathcal{R}|_A$.
	Since $|A'|\ge2$, these lotteries are distinct.
	Thus, $\{R'\in\mathcal{R}|_A\colon |f(R',A)| = 1\}$ is not dense in $\mathcal{R}|_A$ at $R^x$.
	This contradicts decisiveness of $f$.
\end{proof}

\subsection{$\ml\subseteq f$}

In this section we show that every PSCF $f$ that satisfies population-consistency and composition-consistency has to yield \emph{all} maximal lotteries. To this end, we first prove an auxiliary lemma.
It was shown by \citet{McGa53a} that every complete and anti-symmetric relation is the majority relation of some profile with a bounded number of voters. We show an analogous statement for skew-symmetric matrices and fractional preference profiles. 

\begin{lemma}\label{lem:mcgarvey}
	Let $M\in\mathbb{Q}^{n\times n}$ be a skew-symmetric matrix. Then, there are $R\in\mathcal{R}|_{[n]}$ and $c\in \mathbb{Q}_{>0}$ such that $c M = M_R$.
	Furthermore, if there is $\pi\in\Pi([n])$ such that $M(i,j) = M(\pi(i),\pi(j))$ for all $i,j\in[n]$, then $R = \pi(R)$.
\end{lemma}

\begin{proof}
	For all $i,j\in [n]$ with $i\neq j$, let $R^{ij}\in\mathcal{R}|_{[n]}$ be the profile such that $R^{ij}(\pref) = \nicefrac{1}{(n-1)!}$ if $i\succ j$ and $\{i,j\}$ is a component in $R^{ij}$ and $R^{ij}(\pref) = 0$ otherwise.
	By construction, we have that $M_{R^{ij}}(i,j) = 1$ and $M_{R^{ij}}(x,y) = 0$ for all unordered pairs $\{x,y\}\neq\{i,j\}$.
	Let $c = 1/\sum_{i,j\colon M(i,j)>0} M(i,j)$ and $R = c\sum_{i,j\colon M(i,j)>0} M(i,j) R^{ij}$.
	Then, we have that $M_R = cM$.
	The second part of the lemma follows from the symmetry of the construction.
\end{proof}

For profiles which admit a unique maximal lottery, it follows from \lemref{lem:fml} that $f = \ml$. It turns out that every maximal lottery that is a vertex of the set of maximal lotteries in one of the remaining profiles is the limit point of a sequence of maximal lotteries of a sequence of profiles with a unique maximal lottery converging to the original profile. The proof of \lemref{lem:mlf} heavily relies on the continuity of $f$.

\begin{lemma}\label{lem:mlf}
	Let $f$ be a PSCF that satisfies population-consistency and composition-consistency. Then, $\ml\subseteq f$.
\end{lemma}

\begin{proof}
	Let $f$ be a PSCF that satisfies population-consistency and composition-consistency, $A\in\fone(U)$, and $R\in\mathcal{R}|_A$.
	If follows from \lemref{lem:fml} that $f\subseteq\ml$.
	By neutrality, we can assume without loss of generality that $A = [n]$ and for simplicity $M = M_R$.
	We want to show that $f(R) = \ml(R)$. 
	If $\ml(R)$ is a singleton, it follows from $f\subseteq\ml$ that $f(R) = \ml(R)$.
	Hence, consider the case where $\ml(R)$ is not a singleton.
Let $p\in\ml(R)$ and assume without loss of generality that $\supp(p) = [k]$.

	We first consider the case where $k$ is odd. 
By \lemref{lem:mcgarvey}, there are $S\in\mathcal{R}|_A$ and $c\in\mathbb{Q}_{>0}$ such that
	\[
	M_S = c
		\left(
	      	\begin{array}{cccccc:ccc}
				0 & -\frac{1}{p_1p_2} & 0 & \dots & 0 & \frac{1}{p_kp_1} & 1 & \dots & 1\\
				\frac{1}{p_1p_2} & & & \multicolumn{2}{c}{\multirow{2}{*}{$\ddots$}} & 0 & \multirow{4}{*}{\vdots} & \multirow{4}{*}{$\ddots$} & \multirow{4}{*}{\vdots}\\
				0 & & \multicolumn{2}{c}{\multirow{2}{*}{$\ddots$}} & & \vdots\\
				\vdots & \multicolumn{2}{c}{\multirow{2}{*}{$\ddots$}} & & & 0\\
				0 & & & & & -\frac{1}{p_{k-1}p_{k}}\\
				-\frac{1}{p_{k}p_{1}} & 0 & \dots & 0 & \frac{1}{p_{k-1}p_k} & 0 & 1 & \dots & 1\\
				\hdashline
				-1 & \multicolumn{4}{c}{\dots} & -1 & 0 & \dots & 0\\
				\vdots & \multicolumn{4}{c}{\ddots} & \vdots & \vdots & \ddots & \vdots\\
				-1 & \multicolumn{4}{c}{\dots} & -1 & 0 & \dots & 0\\
	       	\end{array}
		\right)
	\]
	Intuitively, $M_S$ defines a weighted cycle on $[k]$.
	Note that $(p^T M_S)_i = 0$ for all $i\in\supp(p)$ and $(p^T M_S)_i > 0$ for all $i\in A\setminus\supp(p)$, \ie $p$ is a quasi-strict maximin strategy in $M_S$ in the sense of \citet{Hars73b}.
Since $p$ is a maximin strategy in $M_S$, it follows that $p\in\ml(S)$.
	For $\epsilon\in[0,1]$, we define $R^\epsilon = (1-\epsilon) R + \epsilon S$ and $M^\epsilon = M_{R^\epsilon}$.
	Population-consistency implies that $p\in\ml(R^\epsilon)$ for all $\epsilon\in[0,1]$.
Observe that $p$ is a quasi-strict maximin strategy in $M^{\epsilon}$ for every $\epsilon\in(0,1]$.
	Hence, for every maximin strategy $q$ in $M^{\epsilon}$, it follows that $(q^T M^{\epsilon})_i = 0$ for every $i\in[k]$ and $q_i = 0$ for every $i\not\in [k]$.
	It follows from basic linear algebra that
\[\det\left(\left(M_S(i,j)\right)_{i,j\in[k-1]}\right) = c^{k-1}\prod_{i = 1}^{k-1} \left(\frac{1}{p_i}\right)^2\neq 0\text,\]
and hence, $(M_S(i,j))_{i,j\in[k]}$ has rank at least $k-1$.
	More precisely, $(M_S(i,j))_{i,j\in[k]}$ has rank $k-1$, since skew-symmetric matrices of odd size cannot have full rank.\footnote{\label{note:1}A skew-symmetric matrix $M$ of odd size cannot have full rank, since $\det(M) = \det(M^T) = \det(-M) = (-1)^n\det(M) = -\det(M)$ and, hence, $\det(M) = 0$.}
	Furthermore, $\det((M^{\epsilon}(i,j))_{i,j\in[k-1]})$ is a polynomial in $\epsilon$ of order at most $k-1$ and hence, has at most $k-1$ zeros.
	Thus, we can find a sequence $(\epsilon_l)_{l\in\mathbb{N}}$ which converges to zero such that $(M^{\epsilon_l}(i,j))_{i,j\in[k]}$ has rank $k-1$ for all $l\in\mathbb{N}$.
	In particular, if $(q^TM^{\epsilon})_i = 0$ for all $i\in[k]$, then $q = p$.
	This implies that $p$ is the unique maximin strategy in $M^{\epsilon_l}$ for all $l\in\mathbb{N}$ and hence, $\{p\} = \ml(R^{\epsilon_l}) \subseteq f(R^{\epsilon_l})$ for all $l\in\mathbb{N}$.
	It follows from continuity of $f$ that $p\in f(R)$.
	
	Now we consider the case where $k$ is even. 
	$\ml(R)$ is a polytope because it is the solution space of a linear feasibility program.
	Assume that $p$ is a vertex of $\ml(R)$.
	We first show that $p$ is not quasi-strict.\footnote{The proof of this statement does not make use of the fact that $k$ is even and therefore also holds (but is not needed) for odd $k$.}
	Assume for contradiction that $p$ is quasi-strict, \ie $(p^T M)_i > 0$ for all $i\not\in [k]$. 
	Then, $\supp(q)\subseteq [k]$ for every maximin strategy $q$ of $M$.
	But then $(1+\epsilon) p - \epsilon q$ is also a maximin strategy in $M$ for small $\epsilon > 0$ as $p$ is a quasi-strict maximin strategy in $M$.
	This contradicts the assumption that $p$ is a vertex of $\ml(R)$.
	
	Hence, we may assume without loss of generality that $(p^T M)_{k+1} = 0$. Let $e_1 = M(k+1,1)/p_2$ and $e_i = (M(k+1,i) + p_{i-1}e_{i-1})/p_{i+1}$ for $i\in\{2,\dots,k-1\}$.
By \lemref{lem:mcgarvey}, there are $S\in\mathcal{R}|_A$ and $c\in\mathbb{Q}_{>0}$ such that
	 
	\[
	M_S = c
		\left(
	      	\begin{array}{ccccc:c:ccc}
				0 & e_1 & 0 & \dots & 0 & 0 & 1 & \dots & 1\\
				 -e_1 & & & \ddots & \vdots & \multirow{3}{*}{\vdots} & \multirow{4}{*}{\vdots} & \multirow{4}{*}{$\ddots$} & \multirow{4}{*}{\vdots}\\
				0 & & \ddots & & 0 & \\
				\vdots & \ddots & & & e_{k-1} & & & &\\
				0 & \dots & 0 &  -e_{k-1} & 0 & 0 & \\
				\cdashline{1-6}
				0 & \multicolumn{3}{c}{\dots} & 0 & 0 & 1 & \dots & 1\\
				\hdashline
				-1 & \multicolumn{4}{c}{\dots} & -1 & 0 & \dots & 0\\
				\vdots & \multicolumn{4}{c}{\ddots} & \vdots & \vdots & \ddots & \vdots\\
				-1 & \multicolumn{4}{c}{\dots} & -1 & 0 & \multicolumn{1}{c}{\dots} & 0\\
	       	\end{array}
		\right)
	\]
	Note that $M_S(1,k) = M_S(k,1) = 0$.
For $\epsilon>0$, let $R^\epsilon = (1-\epsilon) R + \epsilon S$ and $M^\epsilon = M_{R^\epsilon}$.
We claim that $p^\epsilon$ defined as follows is a maximin strategy in $M^{\epsilon}$. To this end, let $s_\epsilon = \frac{\epsilon c}{1 - \epsilon + \epsilon c}$.
	\[p^\epsilon_i =
	\begin{cases}
		(1-s_\epsilon)p_i &\text{if } i\in [k]\text,\\
		s_\epsilon &\text{if } i = k+1\text{, and}\\
		0 &\text{otherwise.}
	\end{cases}
	\]
	Note that $\nicefrac{1}{c}\,(p^TM_S)_1 = -p_2e_1 = -M(k+1,1)$ and, for $i\in\{2,\dots,k-1\}$,
	\[
	\frac{1}{c}(p^TM_S)_i = p_{i-1}e_{i-1} -p_{i+1}e_i = p_{i-1}e_{i-1} - (M(k+1,i) + p_{i-1}e_{i-1}) = -M(k+1,i)\text.
	\]
	To determine $(p^TM_S)_k$, we first prove inductively that $p_ie_i = \nicefrac{1}{p_{i+1}}\,\sum_{j=1}^{i} M(k+1,j)p_j$ for all $i\in[k-1]$. For $i = 1$, this follows from the definition of $e_1$. 
	Now, let $i\in\{2,\dots,k-1\}$. 
	Then,
	\begin{align*}
	p_ie_i &= \frac{p_i}{p_{i+1}}(M(k+1,i) + p_{i-1}e_{i-1}) = \frac{p_i}{p_{i+1}}(M(k+1,i) + \frac{1}{p_{i}}\sum_{j=1}^{i-1} M(k+1,j)p_{j})\\
	&= \frac{1}{p_{i+1}}\sum_{j=1}^{i} M(k+1,j)p_{j}\text,
	\end{align*}
	where the second equality follows from the induction hypothesis. 
	Now,
	\[
	\frac{1}{c}(p^TM_S)_{k} = p_{k-1}e_{k-1} = \frac{1}{p_k} \sum_{j=1}^{k-1} M(k+1,j) p_j = -\frac{1}{p_k} M(k+1,k)p_k = -M(k+1,k)\text,
	\]
	where the third equality follows from the fact that $(p^TM)_{k+1} = 0$.
	
	For $i\in[k]$, it follows from $(p^TM)_i = 0$ that $((p^\epsilon)^TM)_i = s_\epsilon M(k+1,i)$.
	Then, for $i\in[k]$,
	\[((p^\epsilon)^TM^\epsilon)_i = (1-\epsilon)s_\epsilon M(k+1,i) + \epsilon c (1-s_\epsilon) (-M(k+1,i)) = 0\text.\]
	Furthermore, it follows from $(p^TM)_{k+1} = 0$ that $((p^\epsilon)^TM^\epsilon)_{k+1} = 0$ as $M(k+1,k+1) = 0$, and, for $i\in A\setminus[k+1]$,
	\[((p^\epsilon)^TM^\epsilon)_i \ge (1-\epsilon) s_\epsilon M(k+1,i) + \epsilon c \ge -(1-\epsilon) s_\epsilon + \epsilon c > 0\text.\]
This shows that $p^\epsilon$ is a maximin strategy in $M^{\epsilon}$ and hence, $p^\epsilon\in\ml(R^\epsilon)$.
	Since $|\supp(p^\epsilon)|$ is odd, it follows from the first case that $p^\epsilon\in f(R^\epsilon)$.
	Note that $s_\epsilon$ goes to $0$ as $\epsilon$ goes to $0$.
	Hence, $p^\epsilon$ goes to $p$ as $\epsilon$ goes to $0$.
	It now follows from continuity of $f$ that $p\in f(R)$.

	Together, we have that $p\in f(R)$ for every vertex $p$ of $\ml(R)$. Since every lottery in $\ml(R)$ can be written as a convex combination of vertices, convexity of $f(R)$ implies that $f(R) = \ml(R)$.
 \end{proof}

\thmref{thm:fisml} then follows directly from Lemmas~\ref{lem:fml} and~\ref{lem:mlf}.

\begin{theorem}
	A PSCF $f$ satisfies population-consistency and composition-consistency if and only if $f=\ml$.
\end{theorem}

\end{document}